\documentclass[
    aps,
    prb,
    reprint,
    superscriptaddress,
    longbibliography
]{revtex4-2}

\usepackage{mathtools}
\usepackage{amssymb}
\usepackage{amsthm}

\usepackage{newtxtext}
\usepackage{microtype}

\usepackage{xcolor}
    \definecolor{ARegion}{HTML}{87CEEB}
    \definecolor{BRegion}{HTML}{F08080}
    \definecolor{CRegion}{HTML}{FFFF80}
\usepackage{tikz}
\usetikzlibrary{shapes.geometric}

\usepackage{enumitem}

\usepackage[bookmarks=false]{hyperref}
\hypersetup{
    colorlinks=true,
    linktoc=all,
    allcolors=blue
}

\DeclarePairedDelimiter{\abs}{\lvert}{\rvert}
\DeclarePairedDelimiter{\norm}{\lVert}{\rVert}
\DeclarePairedDelimiter{\ket}{\lvert}{\rangle}

\DeclarePairedDelimiterX{\braket}[2]{\langle}{\rangle}{#1\delimsize\vert\mathopen{}#2}
\DeclarePairedDelimiterX{\ketbra}[2]{\lvert}{\rvert}{#1\delimsize\rangle\!\delimsize\langle\mathopen{}#2}
\DeclarePairedDelimiterX{\mel}[3]{\langle}{\rangle}{#1\delimsize\vert\mathopen{}#2\delimsize\vert\mathopen{}#3}
\DeclarePairedDelimiterX{\expval}[1]{\langle}{\rangle}{#1}

\newtheorem{theorem}{Theorem}
\newtheorem{lemma}[theorem]{Lemma}
\newtheorem{proposition}[theorem]{Proposition}
\newtheorem*{restatedtheorem}{Theorem~\ref*{thm:interval-spectrum}}

\DeclareMathOperator{\Ad}{Ad}
\DeclareMathOperator{\Tr}{Tr}
\DeclareMathOperator{\sgn}{sgn}
\DeclareMathOperator{\rank}{rank}
\DeclareMathOperator{\Arg}{Arg}
\renewcommand{\Im}{\operatorname{Im}}

\newcommand{\dd}[1]{\mathop{}\!\mathrm{d}#1}
\newcommand{\defeq}{\equiv}
\NewDocumentCommand{\CMI}{o m m m}{%
    I\IfValueT{#1}{_{#1}}(#2\,\mathord{:}\,#3\vert#4)%
}

\begin{document}

\title{An Instability of the Quantum Markov Length}

\author{David C. Newsom}
\email{dcn@uchicago.edu}
\author{Michael Levin}
\affiliation{Leinweber Institute for Theoretical Physics, University of Chicago, Chicago, Illinois 60637,  USA}

\begin{abstract}
    The Markov length, the length scale governing the exponential decay of conditional mutual information (CMI) in a quantum many-body state, has proven useful for characterizing mixed-state phases of matter.
    Nevertheless, we show that a state's Markov length is not generally stable under weak local channels.
    The failure of stability occurs already at a pure product state.
    We demonstrate this explicitly in a one-dimensional fermionic chain, where a fine-tuned depth-two circuit of nearest-neighbor Gaussian channels maps the vacuum to a state whose CMI decays \emph{algebraically} and whose Markov length is therefore infinite, even when the channels are taken arbitrarily close to the identity.
    The algebraic decay is rooted in the nonanalyticity of the von Neumann entropy at rank-deficient density operators.
    We exhibit an analogous instability on an ordinary qubit chain.
    More generally, we show that the Markov length is unstable in any translation-invariant state with even a single rank-deficient reduced density operator.
    These results pose a challenge to recent definitions of mixed-state phase equivalence based on the Markov length.
\end{abstract}

\maketitle

\textbf{\emph{Introduction}}---Recent efforts to define and characterize mixed-state phases of matter have drawn on analogies with zero-temperature quantum phases.
In the study of zero-temperature phases, the central question is to understand how properties of the ground state or low-energy excitations change as parameters in the underlying Hamiltonian are varied.
In the mixed-state setting, one instead begins with a ground state or a Gibbs state of a many-body Hamiltonian and subjects it to noise modeled by a finite-depth circuit of local channels.
One then asks how the properties of the resulting state change as the parameters of these channels, particularly their strength, are varied~\cite{Fan:DiagnosticsMixedState:2024,Bao:MixedStateTopological:2026,Lee:QuantumCriticality:2023,Zou:ChannelingQuantum:2023,Sang:MixedStateQuantum:2024,Sang:StabilityMixedState:2025,Yang:TopologicalMixed:2026,Sang:MixedStatePhases:2026}.

Within this analogy, a state's ``Markov length''~\cite{Sang:StabilityMixedState:2025} has been proposed to play a role similar to that of a Hamiltonian's energy gap in zero-temperature phases.
To define this quantity, let \(\rho\) be a density operator of a lattice many-body system.
For a simply connected region \(A\) shielded from an exterior region \(C\) by a buffer \(B\) (Fig.~\ref{fig:markov-length-geometry}), consider the conditional mutual information (CMI)
\begin{equation}\label{eq:cmi-definition}
    \CMI{A}{C}{B}_{\rho} \defeq S(\rho_{AB}) + S(\rho_{BC}) - S(\rho_{B}) - S(\rho_{ABC}),
\end{equation}
where \(S(\rho_{R}) \defeq -\Tr [\rho_{R} \log(\rho_{R})]\) is the von Neumann entropy associated with region \(R\).
The Markov length of \(\rho\) is the smallest length scale \(\xi\) for which the CMI admits an exponential bound in the buffer width (up to polynomial factors in the region sizes), i.e.,
\begin{equation}
    \CMI{A}{C}{B}_{\rho} \leq \operatorname{poly}(\abs{A}, \abs{B}, \abs{C}) e^{-\operatorname{dist}(A, C)/\xi},
\end{equation}
uniformly over all such choices of \(A\), \(B\), and \(C\).
If no finite \(\xi\) satisfies this bound, indicating slower-than-exponential decay of the CMI, the state \(\rho\) is said to have infinite Markov length.

The operational significance of the Markov length is that it provides a criterion for when local-channel evolution on a state \(\rho\) can be \emph{locally} reversed.
More precisely, let \(\mathcal{E}\) be a finite-depth circuit of local channels.
Ref.~\cite{Sang:StabilityMixedState:2025} showed that if the Markov length remains bounded throughout the application of \(\mathcal{E}\) \footnote{More precisely, the Markov length must remain bounded after applying any number of complete layers of \(\mathcal{E}\), followed by an arbitrary subset of channels in the next layer.}, then there exists another finite-depth circuit of local channels that undoes the effect of \(\mathcal{E}\) on \(\rho\), one gate at a time.
Subsequent work~\cite{Sang:MixedStatePhases:2026} showed that states connected by such ``locally reversible'' circuits share many long-distance qualitative properties, including topological degeneracy and the structure of anyon-like excitations.
Therefore, just as a nonvanishing energy gap ensures that the long-distance properties of ground states are preserved along a path of Hamiltonians~\cite{Hastings:QuasiadiabaticContinuation:2005}, a bounded Markov length ensures that long-distance properties are preserved along a path of states.

\begin{figure}
    \centering
    \input{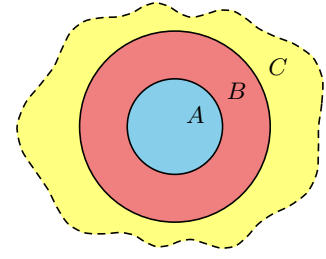}
    \caption{\label{fig:markov-length-geometry}
        Geometry used to define the Markov length in two dimensions.
        Here, \(C\) is any finite region encircling \(A\), with \(B\) filling the space between them.
    }
\end{figure}

This analogy naturally raises the question of whether the Markov length is stable under small perturbations: if \(\rho_{0}\) has finite Markov length, does \(\mathcal{E}(\rho_{0})\) also have finite Markov length for any finite-depth circuit \(\mathcal{E}\) of sufficiently weak local channels \footnote{A variant of this question was originally posed in Ref.~\cite{Sang:StabilityMixedState:2025} in the form of a conjecture.}?
An affirmative answer would provide a simple stability principle for mixed-state phases, implying that the long-distance qualitative properties of any finite-Markov-length state are robust to weak local noise and decoherence.
Several results support such a conclusion.
For example, the Markov length remains finite up to a nonzero noise strength in toric-code models subject to Pauli noise~\cite{Sang:StabilityMixedState:2025,Yang:TopologicalMixed:2026,Vijay:InformationCritical:2026}.
Analytical results have also established stability for certain finite-temperature Gibbs states of commuting Hamiltonians under restricted classes of local channels~\cite{Zhang:StabilityMixedstate:2025,Zhang:ApproximateMarkovianity:2026}.

In this work, we show that the answer to this question is negative in general.
The failure of stability occurs already at a pure product state \(\rho_{0}\), which has identically vanishing CMI and zero Markov length.
More specifically, we construct circuits of fixed finite depth whose constituent local channels can be made arbitrarily close to the identity, yet for which \(\mathcal{E}(\rho_{0})\) has CMI that decays \emph{algebraically} with the buffer width and hence has \emph{infinite} Markov length.

\textbf{\emph{Fermionic example}}---The simplest example exhibiting this instability is a one-dimensional chain of spinless \emph{fermions} with annihilation operators \(\{c_{j}\}_{j \in \mathbb{Z}}\) (we will discuss bosonic variants of this example below).

We take the initial state \(\rho_{0}\) to be the unoccupied many-body vacuum \(\ketbra{\Omega}{\Omega}\) and consider a depth-two circuit \(\mathcal{E}\) of nearest-neighbor Gaussian channels \(\mathcal{E}_{j}\) (Fig.~\ref{fig:depth-two-channel}).
That is, \(\mathcal{E} = \mathcal{E}^{(1)} \circ \mathcal{E}^{(0)}\), where
\begin{equation}\label{eq:channel-layers}
    \mathcal{E}^{(0)} = \prod_{j \in 2\mathbb{Z}} \mathcal{E}_{j}, \qquad \mathcal{E}^{(1)} = \prod_{j \in 2\mathbb{Z} + 1} \mathcal{E}_{j}.
\end{equation}
We define each channel \(\mathcal{E}_{j}\), acting on sites \(j\) and \(j + 1\), in three steps.
First, we introduce a fermionic ancilla \(d_{j}\), initialized in the maximally mixed state.
Second, we couple the ancilla to the two-site mode
\begin{equation}\label{eq:two-site-mode}
    s_{j} \defeq \frac{1}{\sqrt{2}} (c_{j} + c_{j+1}),
\end{equation}
by applying the Gaussian unitary
\begin{equation}\label{eq:system-ancilla-unitary}
    U_{j} = e^{\epsilon(s_{j}^{\dagger} d_{j} - d_{j}^{\dagger} s_{j})}.
\end{equation}
Finally, we trace out the ancilla.
Altogether, the channel \(\mathcal{E}_{j}\) is given by
\begin{equation}\label{eq:local-channel}
    \mathcal{E}_{j}(\rho) = \Tr_{d_{j}}\Bigl[U_{j} \bigl(\rho \otimes I_{d_{j}}/2\bigr) U_{j}^{\dagger}\Bigr].
\end{equation}
Here, the real parameter \(\epsilon\) controls the strength of \(\mathcal{E}_{j}\).
In what follows, we take \(0 \leq \epsilon < \frac{\pi}{2}\) \footnote{
    The restriction \(0 \leq \epsilon < \frac{\pi}{2}\) is imposed only for simplicity.
    Relaxing it requires several technical steps in the appendix to be modified depending on the sign of \(\cos(\epsilon)\), without changing the qualitative behavior of the CMI.
}.

Let \(\rho = \mathcal{E}(\rho_{0})\) be the state produced by this circuit.
The essential feature of \(\rho\), as we will establish below, is that the conditional mutual information \(\CMI{A}{C}{B}_{\rho}\) decays \emph{algebraically} with the width of the buffer \(B\) for every nonzero \(\epsilon\), with \(A\), \(B\), and \(C\) arranged as in Fig.~\ref{fig:depth-two-channel}.
The Markov length is therefore zero at \(\epsilon = 0\) and infinite for every \(\epsilon > 0\) \footnote{
    Strictly speaking, infinite Markov length is defined using the two-sided geometry \(C\!-\!B\!-\!A\!-\!B\!-\!C\), the direct one-dimensional analog of Fig.~\ref{fig:markov-length-geometry}, rather than the one-sided geometry \(A\!-\!B\!-\!C\) used here.
    For our state, however, this distinction is immaterial: finite-depth factorization and data processing show that algebraic decay in the one-sided geometry is sufficient to imply algebraic decay in the two-sided geometry.
}, thereby exhibiting the claimed instability.

To compute \(\CMI{A}{C}{B}_{\rho}\), we use the fact that \(\rho\) is \emph{Gaussian}~\cite{Bravyi:LagrangianRepresentation:2005}: the initial state \(\rho_{0}\) is Gaussian, and Gaussianity is preserved by each channel \(\mathcal{E}_{j}\).
The state \(\rho\) is therefore completely described by its correlation matrix,
\begin{equation}\label{eq:correlation-matrix}
    C_{ij} \defeq \expval{c_{i}^{\dagger} c_{j}}_{\rho} = \Tr\bigl[\rho \, c_{i}^{\dagger} c_{j}\bigr].
\end{equation}
In particular, the von Neumann entropy associated with a region \(R\) is given by~\cite{Dierckx:FermionicQuasifree:2008}
\begin{equation}\label{eq:gaussian-entropy}
    S(\rho_{R}) = \Tr h\bigl(C_{R}\bigr),
\end{equation}
where \(C_{R} = \bigl(C_{ij}\bigr)_{i,j \in R}\) denotes the restricted correlation matrix, and \(h(x) \defeq -x \log(x) - (1 - x) \log(1 - x)\).

\begin{figure}
    \centering
    \input{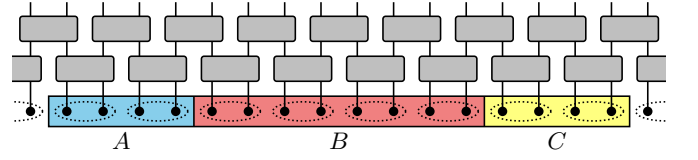}
    \caption{\label{fig:depth-two-channel}
        Gate structure of the depth-two channel \(\mathcal{E}\) (gray boxes).
        Individual sites (black dots) are grouped into two-site unit cells (dotted lines).
        The CMI regions \(A\), \(B\), and \(C\) each consist of a whole number of unit cells.
    }
\end{figure}

Let us compute the correlation matrix \eqref{eq:correlation-matrix}.
The two-site translation invariance of the circuit is inherited by \(\rho\) and its correlation matrix, so that \(C_{ij} = C_{i+2\ell,j+2\ell}\).
We thus group the chain into two-site unit cells indexed by \(x \in \mathbb{Z}\) (Fig.~\ref{fig:depth-two-channel}), with unit cell \(x\) consisting of the neighboring sites \((2x, 2x+1)\).
This partitions the correlation matrix into \(2 \times 2\) blocks \(C_{x-y}\), whose components are \((C_{x-y})_{\mu\nu} = C_{2x+\mu,2y+\nu}\) for \(\mu, \nu \in \{0, 1\}\).
These blocks are compactly encoded by the \(k\)-space matrix \(C(k) \defeq \sum_{x \in \mathbb{Z}} e^{i k x} C_{x}\).
A straightforward calculation gives (see Appendix~\ref{app:correlation-matrix})
\begin{equation}\label{eq:correlation-symbol}
    C(k) = \begin{pmatrix}
        a(k) & b(k) \\[1ex]
        \bar{b}(k) & a(k)
    \end{pmatrix},
\end{equation}
where
\begin{align}
    \label{eq:correlation-symbol-diagonal} a(k) &= \frac{1}{2} \sin^{2}(\epsilon) - \frac{1}{4} \sin^{4}(\epsilon) \cos^{2}\Bigl(\frac{k}{2}\Bigr), \\[0.5ex]
    \label{eq:correlation-symbol-offdiagonal} b(k) &= \frac{1}{4} \sin^{2}(\epsilon) (1 + e^{i k}) \Bigl(\cos^{4}\Bigl(\frac{\epsilon}{2}\Bigr) + \sin^{4}\Bigl(\frac{\epsilon}{2}\Bigr) e^{i k}\Bigr).
\end{align}
Fourier transforming back to real space shows that \(C_{x-y} = 0\) whenever \(\abs{x-y} > 2\).
The correlation matrix is therefore \emph{finite}-range, as expected for a state obtained from a product state by a finite-depth circuit of local channels.

The spectrum of the correlation matrix consists of two bands obtained by diagonalizing \(C(k)\),
\begin{equation}\label{eq:correlation-eigenvalue-bands}
    \lambda_{\pm}(k) = a(k) \pm \abs{b(k)}.
\end{equation}
The crucial property of our example is that the lower band \(\lambda_{-}\) vanishes at \(k = 0\):
\begin{equation}\label{eq:lower-band-quadratic-zero}
    \lambda_{-}(k) = \gamma k^{2} + O(k^{4}), \qquad \gamma \defeq \frac{1}{8} \frac{\sin^{2}(\epsilon)}{1 + \cos^{2}(\epsilon)}.
\end{equation}

This zero has a simple physical origin: the channel \(\mathcal{E}\) has a dark mode,
\begin{equation}
    c_{\pi} = \sum_{j} (-1)^{j} c_{j}.
\end{equation}
Indeed, \(c_{\pi}\) is orthogonal to each of the modes \(s_{j}\) in \eqref{eq:two-site-mode} to which the ancillas are coupled.
It is therefore unaffected by \(\mathcal{E}\) and remains unoccupied in the state \(\rho\).
Consequently, the weights \(v_{j} = (-1)^{j}\) of this dark mode form a null vector of the correlation matrix \(C\).
In \(k\)-space, this null vector corresponds precisely to the \(k = 0\) point of the lower band \(\lambda_{-}\).
We note that a closely related free-fermion model with an analogous dark mode was studied in Ref.~\cite{Pocklington:EfficientSimulation:2025}.

We now show how the vanishing of \(\lambda_{-}\) at \(k = 0\) gives rise to algebraically decaying \(\CMI{A}{C}{B}_{\rho}\).
At the heart of our computation is an asymptotic formula for the von Neumann entropy of a large interval.
Specifically, for an interval containing \(N\) unit cells, we will show that
\begin{equation}\label{eq:interval-entropy-asymptotics}
    S(\rho_{\mathrm{int}}) = N \bar{s} + (\text{const.}) - \frac{c}{N^{2}} + O(N^{-3}),
\end{equation}
where \(\bar{s}\) and \(c\) are constants defined below.
The first term describes the usual volume-law scaling of the entropy, while the second is a constant contribution from the endpoints of the interval.
We will see that the distinctive \(N^{-2}\) term arises directly from the zero of \(\lambda_{-}\).

Substituting \eqref{eq:interval-entropy-asymptotics} into the definition of the CMI \eqref{eq:cmi-definition}, the volume-law and constant terms cancel.
Therefore, for large intervals \(A\), \(B\), and \(C\) containing \(n_{A}\), \(n_{B}\), and \(n_{C}\) unit cells, respectively (Fig.~\ref{fig:depth-two-channel}), the leading term in the CMI is
\begin{equation}\label{eq:cmi-asymptotics}
    \CMI{A}{C}{B}_{\rho} = c \Bigl(\frac{1}{n_{B}^{2}} + \frac{1}{n_{ABC}^{2}} - \frac{1}{n_{AB}^{2}} - \frac{1}{n_{BC}^{2}}\Bigr),
\end{equation}
where \(n_{AB} \defeq n_{A} + n_{B}\), \(n_{BC} \defeq n_{B} + n_{C}\), and \(n_{ABC} \defeq n_{A} + n_{B} + n_{C}\).
Equation~\eqref{eq:cmi-asymptotics} establishes our claim: \(\CMI{A}{C}{B}_{\rho}\) decays \emph{algebraically} with the buffer width \(n_{B}\), so in particular \(\rho\) has infinite Markov length.

It remains to derive the entropy formula \eqref{eq:interval-entropy-asymptotics}.
As a warm-up, we first compute the entropy \(S(\rho_{\mathrm{ring}})\) for a \emph{ring} of \(N\) unit cells.
The ring geometry is illuminating because its entropy contains a similar \(N^{-2}\) term to that in \eqref{eq:interval-entropy-asymptotics}, but the calculation is simpler.

By \eqref{eq:gaussian-entropy}, the computation of \(S(\rho_{\mathrm{ring}})\) reduces to summing \(h\) over the eigenvalues of \(C_{\mathrm{ring}}\).
These eigenvalues are given by the same bands \eqref{eq:correlation-eigenvalue-bands} as on the infinite line, but with \(k\) quantized in multiples of \(\frac{2\pi}{N}\).
The entropy can therefore be written as
\begin{equation}\label{eq:ring-entropy-momentum-sum}
    S(\rho_{\mathrm{ring}}) = \sum_{j = 0}^{N-1} g\Bigl(\frac{2\pi j}{N}\Bigr),
\end{equation}
where \(g(k) \defeq h(\lambda_{-}(k)) + h(\lambda_{+}(k))\).

Since \(g\) is \(2\pi\)-periodic, it has the Fourier expansion \(g(k) = \sum_{m \in \mathbb{Z}} \hat{g}(m) e^{i m k}\), where \(\hat{g}(m) = \int_{-\pi}^{\pi} \frac{\dd{k}}{2\pi} g(k) e^{-i m k}\).
Substituting this expansion into \eqref{eq:ring-entropy-momentum-sum}, only Fourier modes whose mode number \(m\) is a multiple of \(N\) survive:
\begin{equation}\label{eq:ring-entropy-fourier-sum}
    S(\rho_{\mathrm{ring}}) = N \hat{g}(0) + N \sum_{j \neq 0} \hat{g}(jN).
\end{equation}
Here, we have separated out the \(m = 0\) mode because it gives the leading extensive contribution to the entropy.

The asymptotic behavior of the sum in \eqref{eq:ring-entropy-fourier-sum} is determined by the decay of the Fourier coefficients \(\hat{g}(m)\), which in turn is controlled by the differentiability of \(g\).
While the Fourier coefficients of smooth functions necessarily decay superpolynomially with \(\abs{m}\), the function \(g\) is \emph{not} smooth: its second derivative has a logarithmic singularity arising from the zero of \(\lambda_{-}\).
Specifically, using \eqref{eq:lower-band-quadratic-zero}, we find
\begin{equation}
    g^{\prime \prime}(k) = -2\gamma \log(k^{2}) + O(1) \to \infty
\end{equation}
as \(k \to 0\).
This singularity causes the Fourier coefficients to decay only \emph{algebraically} (see Appendix~\ref{app:fourier-asymptotics}):
\begin{equation}\label{eq:entropy-fourier-asymptotics}
    \hat{g}(m) = -\frac{2\gamma}{\abs{m}^{3}} + O(\abs{m}^{-5}).
\end{equation}
Substituting \eqref{eq:entropy-fourier-asymptotics} into \eqref{eq:ring-entropy-fourier-sum} gives
\begin{equation}\label{eq:ring-entropy-asymptotics}
    S(\rho_{\mathrm{ring}}) = N \bar{s} - \frac{4\gamma \zeta(3)}{N^{2}} + O(N^{-4}),
\end{equation}
where \(\bar{s} \defeq \hat{g}(0) = \int_{-\pi}^{\pi} \frac{\dd{k}}{2\pi} g(k)\) and \(\zeta(l) \defeq \sum_{n = 1}^{\infty} \frac{1}{n^{l}}\) is the Riemann zeta function.

With our warm-up complete, we return to the interval entropy \eqref{eq:interval-entropy-asymptotics}.
Here, we give a heuristic derivation of this result, leaving a fully rigorous treatment to Appendix~\ref{app:interval-entropy-asymptotics}.

On physical grounds, we expect a large interval and a large ring to have the same limiting entropy \emph{density}.
The ring calculation therefore fixes the leading extensive term in \eqref{eq:interval-entropy-asymptotics} to be \(N \bar{s}\).
Because the correlation matrix is finite range, we also expect any entropy contributions arising from the endpoints to be local and therefore independent of the interval length.
This justifies the constant term in \eqref{eq:interval-entropy-asymptotics}.

Any remaining terms should be finite-size corrections that vanish as \(N \to \infty\).
Based on the ring calculation, we expect the leading correction to arise from summing \(h\) over the low-lying eigenvalues of \(C_{\mathrm{int}}\).
For large \(N\), these eigenvalues are asymptotically distributed as (see Appendix~\ref{app:interval-eigenvalues})
\begin{equation}\label{eq:interval-low-lying-spectrum}
    \biggl\{\lambda_{-}\Bigl(\frac{\pi j}{N}\Bigr), \quad j = 1, 2, \ldots \biggr\}.
\end{equation}
This motivates introducing the sum
\begin{equation}\label{eq:interval-entropy-proxy-sum}
    G_{N} \defeq \frac{1}{2} \sum_{j = 0}^{2N-1} g\Bigl(\frac{\pi j}{N}\Bigr).
\end{equation}
Using the evenness and periodicity of \(g(k)\), we can equivalently write \(G_{N} = \sum_{j = 1}^{N-1} g(\frac{\pi j}{N}) + \frac{g(0) + g(\pi)}{2}\).
This expression shows that \(G_{N}\) reproduces the sum of \(h\) over the low-lying eigenvalues in \eqref{eq:interval-low-lying-spectrum}, together with additional nonsingular contributions such as \(\sum_{j = 1}^{N-1} h(\lambda_{+}(\frac{\pi j}{N}))\).
We therefore expect \(G_{N}\) to have the same leading algebraic correction as \(S(\rho_{\mathrm{int}})\).
We can read off this correction directly from the ring calculation: comparing \eqref{eq:interval-entropy-proxy-sum} with \eqref{eq:ring-entropy-momentum-sum} shows that \(G_{N}\) is half the entropy of a ring with \(2N\) unit cells.
Applying \eqref{eq:ring-entropy-asymptotics} with \(N\) replaced by \(2N\), we find that the leading algebraic correction to \(G_{N}\) is \(-\frac{\gamma \zeta(3)}{2N^{2}}\).
This gives \eqref{eq:interval-entropy-asymptotics} with \(c = \frac{\gamma \zeta(3)}{2}\), completing the derivation.

\textbf{\emph{Fine-tuning}}---It is worth emphasizing that the algebraic decay of \(\CMI{A}{C}{B}_{\rho}\) in our example requires fine-tuning of both the local channels \(\mathcal{E}_{j}\) and the initial state \(\rho_{0}\).
For example, consider varying the relative weights in \eqref{eq:two-site-mode} by taking \(s_{j} = \cos(\theta) c_{j} + \sin(\theta) c_{j+1}\) with \(0 < \theta < \frac{\pi}{2}\).
The resulting state \(\rho = \mathcal{E}(\rho_{0})\) is again described by a \(2 \times 2\) momentum-space correlation matrix, but for \(\theta \neq \frac{\pi}{4}\), both eigenvalue bands \(\lambda_{\pm}(k)\) are bounded away from zero (see Appendix~\ref{app:fine-tuning}).

The same behavior occurs if we instead leave the channels unchanged but replace the initial vacuum by the finite-temperature Gibbs state \(\rho_{\beta} \propto \exp(-\beta \sum_{j \in \mathbb{Z}} c_{j}^{\dagger} c_{j})\).
This is still a Gaussian product state, but each fermionic mode now has mean occupancy \((1 + e^{\beta})^{-1}\).
A direct calculation of the correlation matrix of the state \(\mathcal{E}(\rho_{\beta})\) shows that its eigenvalue bands are shifted from the zero-temperature bands \eqref{eq:correlation-eigenvalue-bands} according to \(\lambda_{\pm,\beta}(k) = (1 + e^{\beta})^{-1} + \tanh(\frac{\beta}{2}) \lambda_{\pm}(k)\).
This shift leaves both bands bounded away from zero for any finite initial temperature.

In either case, lifting the zero of \(\lambda_{-}\) makes the function \(g(k) = h(\lambda_{-}(k)) + h(\lambda_{+}(k))\) \emph{analytic}, so its Fourier coefficients decay \emph{exponentially} with the mode number.
Consequently, by \eqref{eq:ring-entropy-fourier-sum}, the von Neumann entropy of a ring takes the form \(S(\rho_{\mathrm{ring}}) = N \bar{s} + O(e^{-(\text{const.}) N})\), with no algebraic correction.
Similarly, the interval entropy takes the form \(S(\rho_{\mathrm{int}}) = N \bar{s} + (\text{const.}) + O(e^{-(\text{const.}) N})\), leading to exponentially decaying \(\CMI{A}{C}{B}_{\rho}\).

\textbf{\emph{R\'enyi CMI}}---The role of nonanalyticity in producing algebraic decay is further illustrated by the order-\(n\) R\'enyi CMI \(\CMI[n]{A}{C}{B}_{\rho}\), defined by replacing each von Neumann entropy in \eqref{eq:cmi-definition} with its R\'enyi analog, \(S_{n}(\rho) \defeq \frac{1}{1 - n} \log(\Tr(\rho^{n}))\), where \(n > 0\) and \(n \neq 1\).
The case \(n = 2\), in particular, has recently been studied in Refs.~\cite{Zhang:StrongtoweakSpontaneous:2025,Kuno:RenyiMarkov:2026}.
In our example, the R\'enyi CMI \(\CMI[n]{A}{C}{B}_{\rho}\) decays \emph{exponentially} with \(n_{B}\) for every \emph{integer} \(n \geq 2\), in contrast with the von Neumann CMI.
To see this, recall that the R\'enyi entropy of a fermionic Gaussian state can be computed in the same way as the von Neumann entropy, with the function \(h(x)\) in \eqref{eq:gaussian-entropy} replaced by \(h_{n}(x) \defeq \frac{1}{1 - n} \log(x^{n} + (1 - x)^{n})\)~\cite{Dierckx:FermionicQuasifree:2008}.
Crucially, for integer \(n \geq 2\), \(h_{n}(x)\) is analytic in a neighborhood of \([0, 1]\), including the endpoints.
Thus, despite the zero of \(\lambda_{-}\), the function \(g_{n}(k) \defeq h_{n}(\lambda_{-}(k)) + h_{n}(\lambda_{+}(k))\) remains analytic.
Its Fourier coefficients \(\hat{g}_{n}(m)\) therefore decay exponentially as \(\abs{m} \to \infty\), so by \eqref{eq:ring-entropy-fourier-sum} the R\'enyi entropy of a ring takes the form \(S_{n}(\rho_{\mathrm{ring}}) = N \bar{s}_{n} + O(e^{-(\text{const.}) N})\), with no algebraic correction.
The entropy of an interval has a similar form, but with an additional constant term, leading to exponentially decaying R\'enyi CMI \(\CMI[n]{A}{C}{B}_{\rho}\).

For noninteger \(n\), by contrast, \(h_{n}(x)\) contains a nonanalytic term proportional to \(x^{n}\) near \(x = 0\), so the algebraic correction associated with the zero of \(\lambda_{-}\) reappears (see Appendix~\ref{app:interval-entropy-asymptotics}):
\begin{equation}\label{eq:interval-renyi-entropy-asymptotics}
    S_{n}(\rho_{\mathrm{int}}) = N \bar{s}_{n} + (\text{const.}) - \frac{c(n)}{N^{2n}} + o(N^{-2n}),
\end{equation}
where \(c(n) = \frac{(\pi^{2} \gamma)^{n} \zeta(-2n)}{n - 1}\).
Consequently, the R\'enyi CMI has the same form as in \eqref{eq:cmi-asymptotics}, but with exponent \(2n\) instead of \(2\) and prefactor \(c(n)\) instead of \(c\).
For integer \(n \geq 2\), \(\zeta(-2n) = 0\), so \(c(n)\) vanishes, consistent with the analyticity argument above.
Moreover, the limit \(n \to 1\) gives \(\lim_{n \to 1} c(n) = -2\pi^{2} \gamma \, \zeta'(-2) = \frac{\gamma \zeta(3)}{2}\), recovering the von Neumann result.

We now show that the instability of the Markov length persists well beyond the fermionic finite-depth setting considered above.

\textbf{\emph{Lindbladian example}}---First, the same phenomenon occurs under evolution generated by a time-independent local Lindbladian.
To construct such an example, note that each local gate \eqref{eq:local-channel} can be written exactly as \(\mathcal{E}_{j} = e^{t \mathcal{L}_{j}}\), where \(\mathcal{L}_{j}(\rho) = \frac{1}{2} (s_{j}^{\dagger} \rho s_{j} + s_{j} \rho s_{j}^{\dagger} - \rho)\) and \(t\) is defined by \(e^{-t} = \cos^{2}(\epsilon)\) (see Appendix~\ref{app:channel-representations}).
This motivates considering the channel \(\mathcal{N} = e^{t \mathcal{L}}\), where \(\mathcal{L} = \sum_{j \in \mathbb{Z}} \mathcal{L}_{j}\).
The state \(\rho = \mathcal{N}(\rho_{0})\) is now invariant under translations by a \emph{single} site, and a straightforward calculation shows that it has the \emph{scalar} momentum-space correlation matrix \(C(k) = \frac{1}{2} [1 - \exp(-2t \cos^{2}(\frac{k}{2}))]\).
As in our original model, \(C(k)\) touches zero, now at \(k = \pi\).
A similar analysis gives the CMI in the form \eqref{eq:cmi-asymptotics}, with \(c = \frac{t \zeta(3)}{32}\) (see Appendix~\ref{app:lindbladian-model}).

\textbf{\emph{Bosonic example}}---The mechanism underlying the instability is not specific to fermionic statistics: a closely analogous construction exists for canonical bosons.
In this variant, the \(\{c_{j}\}\) are bosonic annihilation operators obeying the canonical commutation relations \([c_{i}, c_{j}^{\dagger}] = \delta_{ij}\) and \([c_{i}, c_{j}] = 0\).
Each \(\mathcal{E}_{j}\) is a two-site Gaussian bosonic channel defined in the same way as in \eqref{eq:two-site-mode}--\eqref{eq:local-channel}, with the only difference being that the ancillas \(d_{j}\) are now bosons initialized in a Gaussian Gibbs state \(\tau_{d_{j}} \propto e^{-\beta d_{j}^{\dagger} d_{j}}\).
Choosing \(\beta\) so that \(\expval{d_{j}^{\dagger} d_{j}}_{\tau_{d_{j}}} = \frac{1}{2}\), the correlation matrix \eqref{eq:correlation-matrix} is unchanged from the fermionic case \eqref{eq:correlation-symbol}.
We can therefore compute the von Neumann entropy in the same way as in the fermionic case, except that the function \(h(x)\) in \eqref{eq:gaussian-entropy} is replaced by \(h_{\mathrm{b}}(x) \defeq (x + 1) \log(x + 1) - x \log(x)\)~\cite{Holevo:EvaluatingCapacities:2001}.
This modification does not affect the logarithmic singularity at \(k = 0\), since \(h_{\mathrm{b}}(\lambda_{-}(k)) + h_{\mathrm{b}}(\lambda_{+}(k))\) has the same leading nonanalytic behavior as in the fermionic case.
Consequently, the CMI has the same algebraic decay as in \eqref{eq:cmi-asymptotics}, with the same constant \(c\).

\textbf{\emph{Qubit example}}---The examples considered so far are all Gaussian models of canonical bosons or fermions, so it is natural to ask whether the same instability can occur on an ordinary spin chain.
A naive approach is to use the Jordan--Wigner transformation to map the fermionic circuit onto a chain of \emph{qubits}.
This works at the level of states: the Jordan--Wigner image of the fermionic state \(\rho=\mathcal{E}(\rho_{0})\) is a valid qubit state \(\rho^{\mathrm{q}}\) on the infinite chain.
Moreover, for any contiguous interval, the Jordan--Wigner equivalence descends to the corresponding reduced states \(\rho_{\mathrm{int}}\) and \(\rho_{\mathrm{int}}^{\mathrm{q}}\) \footnote{
In general, equivalence under Jordan--Wigner does \emph{not} descend to arbitrary reduced states \(\rho_{R}\) and \(\rho_{R}^{\mathrm{q}}\); see Appendix C of Ref.~\cite{Vidal:QuantumOperations:2021} for a discussion.
It does descend when \(R\) is contiguous with respect to the chosen Jordan--Wigner ordering.},
so they have identical spectra and hence identical entropies.
Since our calculation of the CMI involves only contiguous intervals, \(\rho^{\mathrm{q}}\) has exactly the same algebraically decaying CMI as \(\rho\).

The main subtlety is that while \(\rho^{\mathrm{q}}\) can be prepared from the corresponding qubit product state \(\rho_{0}^{\mathrm{q}}\) by the Jordan--Wigner image of the channel \(\mathcal{E}\), this channel is \emph{nonlocal}, since fermion-parity-odd Kraus operators acquire extended parity strings.
For the particular channel \(\mathcal{E}\), however, one can show that these parity strings do not affect the action of the circuit on the initial product state.
After dropping the strings, the resulting Kraus operators define a local depth-two qubit channel \(\mathcal{E}^{\mathrm{q}}\) whose action on \(\rho_{0}^{\mathrm{q}}\) produces \(\rho^{\mathrm{q}}\).
Since \(\mathcal{E}^{\mathrm{q}}\) becomes arbitrarily weak as \(\epsilon\to0\), this establishes the same instability on an ordinary spin chain.
See Appendix~\ref{app:qubit-model} for the explicit form of \(\mathcal{E}^{\mathrm{q}}\) and further discussion.

\textbf{\emph{Instability beyond product states}}---Finally, this instability is not limited to product states: it extends to any translation-invariant state \(\rho\) with a ``local rank deficiency'', i.e., for which the reduced density operator \(\rho_{R}\) fails to be full rank on at least one finite region \(R\) \footnote{
    The essential requirement is not translation invariance itself, but rather that rank-deficient reduced density operators occur on sufficiently many finite regions throughout the bulk.
    We assume translation invariance as a convenient way to ensure this.
}.
This class includes, for example, the toric-code ground state.

To see this, consider such a state \(\rho\) on a \(d\)-dimensional lattice of \emph{qubits} (analogous arguments apply to qudits and fermions).
By grouping sites into sufficiently large unit cells, we may take the rank-deficient region \(R\) to be a single unit cell satisfying \(\rank \rho_{R} \leq \frac{1}{2} \dim \mathcal{H}_{R}\).
The key idea is that this bound allows us to identify a ``virtual'' qubit within \(R\) that is in a fixed pure state and completely uncorrelated with the rest of the system.
More precisely, there exists a decomposition \(\mathcal{H}_{R} \cong \widetilde{\mathcal{H}}_{R} \otimes \mathbb{C}^{2}\) such that \(\rho_{R} \cong \sigma_{R} \otimes \ketbra{0}{0}\).
Making the same identification in every unit cell, the full state factorizes as \(\rho \cong \sigma \otimes \ketbra{0}{0}^{\otimes \mathbb{Z}^{d}}\).
We can then apply the one-dimensional qubit channel \(\mathcal{E}^{\mathrm{q}}\) independently along parallel chains in this virtual product layer, with each two-qubit gate corresponding to a physical channel acting on two neighboring unit cells.
The resulting state acquires the same algebraically decaying CMI contribution as in the product-state example and therefore has infinite Markov length.

\textbf{\emph{Discussion}}---Our results have implications for the problem of defining a precise notion of mixed-state phase equivalence---that is, specifying when two states \(\rho_{1}\) and \(\rho_{2}\) should be regarded as belonging to the same phase or to distinct phases.
Recently, several definitions have been proposed based on the Markov length and/or local reversibility~\cite{Sang:MixedStatePhases:2026,Yi:UniversalDecay:2026}.
Our results pose a challenge to these definitions.

To see why, consider the definition proposed in Ref.~\cite{Yi:UniversalDecay:2026}, according to which two density operators \(\rho_{1}\) and \(\rho_{2}\) are in the same phase if they are connected in both directions by locally reversible circuits of local channels.
Under this definition, \(\rho = \mathcal{E}(\rho_{0})\) is not in the same phase as the product state \(\rho_{0}\) for any \(\epsilon > 0\).
Indeed, locally reversible circuits preserve exponential decay of the CMI~\cite{Yi:UniversalDecay:2026}, whereas \(\rho_{0}\) has vanishing CMI and \(\rho\) has algebraically decaying CMI.

The definition proposed in Ref.~\cite{Sang:MixedStatePhases:2026} leads to the same conclusion.
There, two density operators are in the same phase if they can be (one-way) connected by a circuit of local channels while the Markov length remains bounded throughout application of the circuit.
Since \(\rho\) itself has infinite Markov length for every \(\epsilon > 0\), no such circuit can connect \(\rho_{0}\) to \(\rho\).

Excluding \(\rho\) from the phase containing \(\rho_{0}\) conflicts with what we regard as a natural requirement for a definition of mixed-state phases: the trivial phase should be a \emph{stable} phase of matter, that is, the product state should remain in the same phase under sufficiently weak local channels.
If one insists on this stability property, then one must conclude that the definitions proposed in Refs.~\cite{Sang:MixedStatePhases:2026,Yi:UniversalDecay:2026} are too restrictive.

We note that the original definition of mixed-state phase equivalence~\cite{Coser:ClassificationPhases:2019,Sang:MixedStateQuantum:2024}, based only on two-way channel connectivity, avoids this stability issue: any state obtained from \(\rho_{0}\) by a local channel circuit remains in the same phase as \(\rho_{0}\).
However, this definition has the shortcoming that it identifies states with distinct topological structure.
For example, it places the product state and the classical loop state in the same phase, even though the latter supports nontrivial topological degeneracy and classical topological memory~\cite{Sang:MixedStatePhases:2026}.
These observations suggest that an appropriate notion of phase equivalence should lie between these extremes: coarser than the criteria of Refs.~\cite{Sang:MixedStatePhases:2026,Yi:UniversalDecay:2026}, but finer than bare two-way connectivity.

One question raised by our work is whether the Markov length can be stable for other choices of the initial state \(\rho_{0}\).
The instability we have demonstrated here relies on \(\rho_{0}\) having at least one rank-deficient reduced density operator.
We therefore cannot rule out the possibility that the Markov length is stable for initial states whose reduced density operators are all full rank.
In particular, the Markov length may be stable for the maximally mixed state or, more generally, for finite-temperature Gibbs states.
The partial stability result of Ref.~\cite{Zhang:StabilityMixedstate:2025} (see also Ref.~\cite{Zhang:ApproximateMarkovianity:2026}) provides evidence in this direction.

Another question is whether the instability of the Markov length found here is intrinsically quantum.
Specifically, consider a lattice of classical spins with finite local state spaces, initially described by a product probability distribution.
Must the distribution produced by any finite-depth circuit of sufficiently weak local stochastic maps have exponentially decaying CMI?
More generally, understanding the stability of the Markov length in classical many-body systems is an interesting direction for future work.

\textbf{\emph{Note added}}---Related results concerning the instability of the Markov length have been obtained independently by Negari, Zou, and Hsieh~\cite{Negari:InstabilityMarkov:2026}.
See also Chen and Grover~\cite{Chen:MarkovLength:2026}, who highlight a tension between RG universality and Markov-length-based definitions of phase equivalence.

\textbf{\emph{Acknowledgments}}---We thank Liang Jiang and Su-un Lee for insightful discussions.
D.C.N. and M.L. acknowledge the support of the Leinweber Institute for Theoretical Physics at the University of Chicago.
This work was supported in part by a Simons Investigator grant (M.L.) and by the Simons Collaboration on Ultra-Quantum Matter, which is a grant from the Simons Foundation (651442, M.L.).

\bibliography{Bibliography}

\clearpage

\onecolumngrid
\appendix

\makeatletter
\renewcommand{\p@subsection}{\thesection.}
\makeatother

\section{Derivation of the correlation matrix}\label{app:correlation-matrix}

In this appendix, we derive the momentum-space correlation matrix for \(\rho\) given in \eqref{eq:correlation-symbol} in the main text.
We begin by determining how the correlation matrix transforms under each two-site channel \(\mathcal{E}_{j}\).
To do so, we need the action of \(U_{j}\) on the fermionic modes \(\{c_{i}\}_{i \in \mathbb{Z}}\).
Recall that, aside from the ancilla \(d_{j}\), the unitary \(U_{j}\) acts nontrivially only on the neighboring modes \(c_{j}\) and \(c_{j+1}\).
With \(\tilde{c}_{i}^{\dagger} = U_{j}^{\dagger} c_{i}^{\dagger} U_{j}\), this action is given by
\begin{equation}\label{eq:unitary-mode-transformation}
    \begin{pmatrix}
        \tilde{c}_{j}^{\dagger} \\[1ex]
        \tilde{c}_{j+1}^{\dagger}
    \end{pmatrix} = X \begin{pmatrix}
        c_{j}^{\dagger} \\[1ex]
        c_{j+1}^{\dagger}
    \end{pmatrix} + L d_{j}^{\dagger},
\end{equation}
where
\begin{equation}
    X = \begin{pmatrix}
        \cos^{2}(\frac{\epsilon}{2}) & -\sin^{2}(\frac{\epsilon}{2}) \\[2ex]
        -\sin^{2}(\frac{\epsilon}{2}) & \cos^{2}(\frac{\epsilon}{2})
    \end{pmatrix}, \qquad L = \frac{1}{\sqrt{2}} \sin(\epsilon) \begin{pmatrix}
        1 \\[1ex]
        1
    \end{pmatrix}.
\end{equation}
Now consider the restriction of the correlation matrix to the sites \((j, j+1)\), with components \(C^{(j,j+1)}_{\mu\nu} = \expval{c_{j+\mu}^{\dagger} c_{j+\nu}}_{\rho}\) for \(\mu, \nu \in \{0, 1\}\).
Under \(\mathcal{E}_{j}\), these components transform as
\begin{equation}
    \tilde{C}^{(j,j+1)}_{\mu\nu} = \expval{c_{j+\mu}^{\dagger} c_{j+\nu}}_{\mathcal{E}_{j}(\rho)} = \Tr\Bigl[\bigl(\rho \otimes I_{d_{j}}/2\bigr) \tilde{c}_{j+\mu}^{\dagger} \tilde{c}_{j+\nu}\Bigr].
\end{equation}
Using \eqref{eq:unitary-mode-transformation} and the fact that \(\expval{d_{j}^{\dagger} d_{j}}_{I_{d_{j}}/2} = \frac{1}{2}\), we obtain
\begin{equation}
    \tilde{C}^{(j,j+1)} = X C^{(j,j+1)} X^{\dagger} + Y, \qquad Y = \frac{1}{2} L L^{\dagger}.
\end{equation}

Taking the product of all two-site channels in a given layer, we find that the full correlation matrix transforms under \(\mathcal{E}^{(l)}\), with \(l \in \{0, 1\}\), as
\begin{equation}\label{eq:layer-correlation-update-real-space}
    C \xrightarrow{\mathcal{E}^{(l)}} X^{(l)} C \bigl(X^{(l)}\bigr)^{\dagger} + Y^{(l)},
\end{equation}
where \(X^{(l)}\) and \(Y^{(l)}\) are the block-diagonal matrices
\begin{equation}
    X^{(l)} = \bigoplus_{j \in 2\mathbb{Z} + l} X^{(j,j+1)}, \qquad Y^{(l)} = \bigoplus_{j \in 2\mathbb{Z} + l} Y^{(j,j+1)}.
\end{equation}
Here, \(X^{(j,j+1)}\) and \(Y^{(j,j+1)}\) denote copies of the \(2 \times 2\) matrices \(X\) and \(Y\) acting on sites \((j, j+1)\).

The matrices \(X^{(l)}\) and \(Y^{(l)}\) are invariant under translations by two sites.
We therefore use two-site unit cells indexed by \(x \in \mathbb{Z}\), with unit cell \(x\) consisting of the neighboring sites \((2x, 2x+1)\).
For any matrix \(M\) with this translation symmetry, the unit-cell decomposition gives \(2 \times 2\) blocks \(M_{x-y}\) with components \((M_{x-y})_{\mu\nu} = M_{2x+\mu,2y+\nu}\) for \(\mu, \nu \in \{0, 1\}\).
These blocks are compactly encoded by the \(k\)-space matrix \(M(k) = \sum_{x \in \mathbb{Z}} e^{i k x} M_{x}\).
Applying this construction to \(X^{(l)}\) and \(Y^{(l)}\) gives
\begin{equation}
    X^{(l)}(k) = \begin{pmatrix}
        \cos^{2}(\frac{\epsilon}{2}) & -\sin^{2}(\frac{\epsilon}{2}) e^{i l k} \\[2ex]
        -\sin^{2}(\frac{\epsilon}{2}) e^{-i l k} & \cos^{2}(\frac{\epsilon}{2})
    \end{pmatrix}, \qquad Y^{(l)}(k) = \frac{1}{4} \sin^{2}(\epsilon) \begin{pmatrix}
        1 & e^{i l k} \\[1ex]
        e^{-i l k} & 1
    \end{pmatrix}.
\end{equation}
The factors \(e^{\pm i k}\) appear in the \(l = 1\) layer because its two-site channels couple neighboring unit cells.

Representing the correlation matrix in \(k\)-space in the same way as \(X^{(l)}\) and \(Y^{(l)}\), the transformation law \eqref{eq:layer-correlation-update-real-space} becomes
\begin{equation}\label{eq:layer-correlation-update-momentum-space}
    C(k) \xrightarrow{\mathcal{E}^{(l)}} X^{(l)}(k) C(k) \bigl(X^{(l)}(k)\bigr)^{\dagger} + Y^{(l)}(k).
\end{equation}
The full channel applies \(\mathcal{E}^{(0)}\) followed by \(\mathcal{E}^{(1)}\).
Since the correlation matrix vanishes for the initial state \(\rho_{0}\), the first layer gives \(C(k) = Y^{(0)}(k)\), while the second gives \(C(k) = X^{(1)}(k) Y^{(0)}(k) (X^{(1)}(k))^{\dagger} + Y^{(1)}(k)\).
Evaluating this expression yields
\begin{equation}\label{eq:correlation-symbol-appendix}
    C(k) = \begin{pmatrix}
        a(k) & b(k) \\[1ex]
        \bar{b}(k) & a(k)
    \end{pmatrix},
\end{equation}
where
\begin{align}
    \label{eq:correlation-symbol-diagonal-appendix} a(k) &= \frac{1}{2} \sin^{2}(\epsilon) - \frac{1}{4} \sin^{4}(\epsilon) \cos^{2}\Bigl(\frac{k}{2}\Bigr), \\[0.5ex]
    \label{eq:correlation-symbol-offdiagonal-appendix} b(k) &= \frac{1}{4} \sin^{2}(\epsilon) (1 + e^{i k}) \Bigl(\cos^{4}\Bigl(\frac{\epsilon}{2}\Bigr) + \sin^{4}\Bigl(\frac{\epsilon}{2}\Bigr) e^{i k}\Bigr).
\end{align}
This is precisely \eqref{eq:correlation-symbol} from the main text.

For completeness, we also write out the correlation matrix in real space.
Taking the Fourier transform of \eqref{eq:correlation-symbol-appendix}, we find that the nonzero real-space correlation blocks are given by
\begin{equation}\label{eq:real-space-correlation-blocks}
    \begin{gathered}
        C_{0} = \frac{1}{2} \sin^{2}(\epsilon) \begin{pmatrix}
            1 - \frac{1}{4} \sin^{2}(\epsilon) & \frac{1}{2} \cos^{4}(\frac{\epsilon}{2}) \\[2ex]
            \frac{1}{2} \cos^{4}(\frac{\epsilon}{2}) & 1 - \frac{1}{4} \sin^{2}(\epsilon)
        \end{pmatrix}, \qquad C_{1} = \frac{1}{8} \sin^{2}(\epsilon) \begin{pmatrix}
            -\frac{1}{2} \sin^{2}(\epsilon) & 1 + \cos^{2}(\epsilon) \\[2ex]
            0 & -\frac{1}{2} \sin^{2}(\epsilon)
        \end{pmatrix}, \\[2ex]
        C_{2} = \frac{1}{4} \sin^{2}(\epsilon) \begin{pmatrix}
            0 & \sin^{4}(\frac{\epsilon}{2}) \\[2ex]
            0 & 0
        \end{pmatrix}.
    \end{gathered}
\end{equation}
The blocks at negative separations are fixed by Hermiticity, \(C_{-x} = C_{x}^{\dagger}\) for \(x = 1, 2\), and all remaining blocks vanish.

An important feature of the real-space formula \eqref{eq:real-space-correlation-blocks} is that \(C\) is \emph{finite-range}.
In particular, using \(C_{2x+\mu,2y+\nu} = (C_{x-y})_{\mu\nu}\), we can see that \(C_{ij} = 0\) for \(\abs{i-j} \geq 4\).
This is expected because \(\rho\) is obtained from a product state by a finite-depth circuit of local channels.

\section{Asymptotic behavior of Fourier coefficients}\label{app:fourier-asymptotics}

In this appendix, we derive the asymptotic formula \eqref{eq:entropy-fourier-asymptotics} for the decay of the Fourier coefficients \(\hat{g}(m)\) as \(\abs{m} \to \infty\).
As discussed in the main text, this result leads immediately to a corresponding asymptotic formula \eqref{eq:ring-entropy-asymptotics} for the von Neumann entropy of a ring.
We also discuss the analogous asymptotic formula related to the order-\(n\) R\'enyi entropy.

\subsection{Review of Fourier asymptotics}\label{sec:fourier-asymptotics-review}

We begin with a brief review of the relationship between the regularity of a periodic function and the asymptotic behavior of its Fourier coefficients.

In general, the asymptotic behavior of a function's Fourier coefficients is closely related to its degree of differentiability.
If \(f\) is a \(p\)-times continuously differentiable periodic function, then repeated integration by parts gives
\begin{equation}\label{eq:fourier-integration-by-parts}
    \hat{f}(m) \defeq \frac{1}{2\pi} \int_{-\pi}^{\pi} f(k) e^{-i m k} \dd{k} = \frac{1}{(i m)^{p}} \, \frac{1}{2\pi} \int_{-\pi}^{\pi} f^{(p)}(k) e^{-i m k} \dd{k}.
\end{equation}
Here the boundary terms vanish because \(f\) and its derivatives are periodic.
Since \(\int_{-\pi}^{\pi} f^{(p)}(k) e^{-i m k} \dd{k} \to 0\) as \(\abs{m} \to \infty\) by the Riemann--Lebesgue lemma, it follows that \(\hat{f}(m) = o(\abs{m}^{-p})\).
If \(f\) is smooth, this holds for every \(p \geq 0\), so its Fourier coefficients decay faster than any power of \(\abs{m}\).
We refer to this as superpolynomial decay and write \(\hat{f}(m) = o(\abs{m}^{-\infty})\).

The decay can be strengthened from superpolynomial to exponential when \(f\) is analytic.
In this case, \(f\) can be analytically continued off the real line to a complex strip \(\{z \in \mathbb{C} : \abs{\Im z} < \delta\}\), whose maximal width is determined by the nearest complex singularities of \(f\).
For any fixed \(0 < \mu < \delta\), periodicity allows the integration contour in \eqref{eq:fourier-integration-by-parts} to be shifted by \(-i \sgn(m) \mu\).
Thus, for \(m \neq 0\),
\begin{equation}
    \hat{f}(m) = \frac{e^{-\mu \abs{m}}}{2\pi} \int_{-\pi}^{\pi} f\bigl(k - i \sgn(m) \mu\bigr) e^{-i m k} \dd{k}.
\end{equation}
The remaining integral tends to zero as \(\abs{m} \to \infty\) by the Riemann--Lebesgue lemma, so \(\hat{f}(m) = o(e^{-\mu \abs{m}})\) for every \(0 < \mu < \delta\).

The entropy functions considered below are not smooth everywhere.
Instead, their algebraic Fourier tails arise from an isolated nonsmooth point.
The following localization lemma shows that only a neighborhood of this point contributes to the power-law asymptotics, while the remainder is superpolynomially small.

\begin{lemma}\label{lem:fourier-localization}
    Let \(f(k)\) be a periodic function that is smooth away from \(k = 0\), and let \(\chi(k)\) be any smooth function compactly supported in \((-\pi, \pi)\) that equals \(1\) in a neighborhood of \(k = 0\).
    Then
    \begin{equation}
        \hat{f}(m) = \frac{1}{2\pi} \int_{-\pi}^{\pi} \chi(k) f(k) e^{-i m k} \dd{k} + o(\abs{m}^{-\infty}).
    \end{equation}
\end{lemma}

\begin{proof}
    Define the remainder \(f_{\mathrm{rem}}(k) \defeq (1 - \chi(k)) f(k)\) for \(k \in [-\pi, \pi]\) and extend it periodically.
    Since \(1 - \chi(k)\) vanishes in a neighborhood of \(k = 0\), the remainder \(f_{\mathrm{rem}}(k)\) is smooth.
    Its Fourier coefficients therefore decay superpolynomially, and hence
    \begin{equation}
        \hat{f}(m) - \frac{1}{2\pi} \int_{-\pi}^{\pi} \chi(k) f(k) e^{-i m k} \dd{k} = \hat{f}_{\mathrm{rem}}(m) = o(\abs{m}^{-\infty}).
    \end{equation}
\end{proof}

\subsection{Asymptotics of \(\hat{g}(m)\)}

We now derive the formula \eqref{eq:entropy-fourier-asymptotics} for the asymptotic behavior of \(\hat{g}(m)\).
Recall that \(g(k)\) is defined by
\begin{equation}\label{eq:entropy-spectral-function}
    g(k) = h\bigl(\lambda_{-}(k)\bigr) + h\bigl(\lambda_{+}(k)\bigr), \qquad h(x) \defeq -x \log(x) - (1 - x) \log(1 - x),
\end{equation}
where \(\lambda_{\pm}(k)\) are given in \eqref{eq:correlation-eigenvalue-bands}.

The first step is to use Lemma~\ref{lem:fourier-localization} to isolate the contribution coming from the logarithmic singularity in \(g(k)\) at \(k = 0\).
Let \(\chi(k)\) be a smooth even function compactly supported in \((-\pi, \pi)\) that equals \(1\) in a neighborhood of \(k = 0\).
Applying Lemma~\ref{lem:fourier-localization} gives
\begin{equation}\label{eq:entropy-fourier-localization}
    \hat{g}(m) = \frac{1}{2\pi} \int_{-\pi}^{\pi} \chi(k) g(k) e^{-i m k} \dd{k} + o(\abs{m}^{-\infty}).
\end{equation}

Next, we use \eqref{eq:entropy-spectral-function} to write \(g(k)\) as a sum of four terms:
\begin{equation}
    g(k) = -\lambda_{-}(k) \log\bigl(\lambda_{-}(k)\bigr) - \bigl(1 - \lambda_{-}(k)\bigr) \log\bigl(1 - \lambda_{-}(k)\bigr) - \lambda_{+}(k) \log\bigl(\lambda_{+}(k)\bigr) - \bigl(1 - \lambda_{+}(k)\bigr) \log\bigl(1 - \lambda_{+}(k)\bigr).
\end{equation}
The last three terms are smooth in the region of support of \(\chi(k)\), and repeated integration by parts shows they only make a superpolynomially small contribution to \(\hat{g}(m)\).
Dropping these terms, we obtain
\begin{equation}
    \hat{g}(m) = -\frac{1}{2\pi} \int_{-\pi}^{\pi} \chi(k) \lambda_{-}(k) \log\bigl(\lambda_{-}(k)\bigr) e^{-i m k} \dd{k} + o(\abs{m}^{-\infty}).
\end{equation}

The function \(\lambda_{-}(k)\) has a quadratic zero at \(k = 0\), with \(\lambda_{-}(k) = \gamma k^{2}+O(k^4)\), so we may write
\begin{equation}
    \chi(k) \lambda_{-}(k) \log\bigl(\lambda_{-}(k)\bigr) = \chi(k) \lambda_{-}(k) \log(k^{2}) + \chi(k) \lambda_{-}(k) \log\Bigl(\frac{\lambda_{-}(k)}{k^{2}}\Bigr).
\end{equation}
The second term is again smooth and can be dropped.
We thus obtain
\begin{align}
    \hat{g}(m) &= -\frac{1}{2\pi} \int_{-\pi}^{\pi} \chi(k) \lambda_{-}(k) \log(k^{2}) e^{-i m k} \dd{k} + o(\abs{m}^{-\infty}) \nonumber \\
    \label{eq:entropy-fourier-logarithmic-integral} &= -\frac{2}{\pi} \int_{0}^{\pi} \chi(k) \lambda_{-}(k) \log(k) \cos(m k) \dd{k} + o(\abs{m}^{-\infty}),
\end{align}
where the second equality follows from the evenness of \(\chi(k)\) and \(\lambda_{-}(k)\).
This integral is of the form considered in Theorem~4 of Ref.~\cite{Erdelyi:AsymptoticExpansions:1956}, which gives
\begin{equation}\label{eq:logarithmic-cosine-integral-asymptotics}
    \int_{0}^{\pi} \chi(k) \lambda_{-}(k) \log(k) \cos(m k) \dd{k} = \frac{\pi}{2} \frac{\lambda^{(2)}_{-}(0)}{\abs{m}^{3}} + O(\abs{m}^{-5}).
\end{equation}
Substituting \eqref{eq:logarithmic-cosine-integral-asymptotics} into \eqref{eq:entropy-fourier-logarithmic-integral} and using \(\lambda^{(2)}_{-}(0) = 2\gamma\), we obtain \eqref{eq:entropy-fourier-asymptotics} from the main text:
\begin{equation}
    \hat{g}(m) = -\frac{2\gamma}{\abs{m}^{3}} + O(\abs{m}^{-5}).
\end{equation}

\subsection{Asymptotics of R\'enyi analog \(\hat{g}_{n}(m)\)}

We now derive the R\'enyi analog of the formula \eqref{eq:entropy-fourier-asymptotics}.
Let \(n\) denote the R\'enyi parameter, which for our purposes can be any positive real number other than \(1\).
Recall that \(g_{n}(k)\) is defined by
\begin{equation}\label{eq:renyi-spectral-function}
    g_{n}(k) = h_{n}\bigl(\lambda_{-}(k)\bigr) + h_{n}\bigl(\lambda_{+}(k)\bigr), \qquad h_{n}(x) \defeq \frac{1}{1 - n} \log\bigl(x^{n} + (1 - x)^{n}\bigr).
\end{equation}
Our goal is to determine the asymptotic decay of the Fourier coefficients \(\hat{g}_{n}(m)\) as \(\abs{m} \to \infty\).
Similarly to the von Neumann case, the decay rate of these Fourier coefficients controls the behavior of the order-\(n\) R\'enyi entropy in a ring geometry, \(S_{n}(\rho_{\mathrm{ring}})\).

First suppose that the R\'enyi parameter \(n\) is an \emph{integer}, i.e.,~\(n \in \{2, 3, \ldots\}\).
In this case, the asymptotic behavior of \(\hat{g}_{n}(m)\) is particularly simple: the function \(g_{n}(k)\) is analytic, so \(\hat{g}_{n}(m)\) decays \emph{exponentially} in \(\abs{m}\), as reviewed in Appendix~\ref{sec:fourier-asymptotics-review}.

Next, consider noninteger \(n\).
In this case \(g_{n}(k)\) has a nonanalyticity at \(k = 0\), so we use Lemma~\ref{lem:fourier-localization} to isolate the \(k = 0\) contribution.
Let \(\chi(k)\) be a smooth even function compactly supported in \((-\pi, \pi)\) that equals \(1\) in a neighborhood of \(k = 0\).
Applying Lemma~\ref{lem:fourier-localization} gives
\begin{align}
    \hat{g}_{n}(m) &= \frac{1}{2\pi} \int_{-\pi}^{\pi} \chi(k) g_{n}(k) e^{-i m k} \dd{k} + o(\abs{m}^{-\infty}) \nonumber \\
    \label{eq:renyi-fourier-localization} &= \frac{1}{2\pi} \int_{-\pi}^{\pi} \chi(k) h_{n}\bigl(\lambda_{-}(k)\bigr) e^{-i m k} \dd{k} + o(\abs{m}^{-\infty}).
\end{align}
Here, the second equality follows from the fact that \(\chi(k) h_{n}(\lambda_{+}(k))\) is a smooth function of \(k\) and therefore makes only a superpolynomially small contribution to the integral.

Next, we use the expansion
\begin{equation}\label{eq:renyi-entropy-series}
    h_{n}(x) = \frac{n}{1 - n} \log(1 - x) + \frac{1}{1 - n} \sum_{p = 1}^{\infty} \frac{(-1)^{p+1}}{p} \biggl(\frac{x}{1 - x}\biggr)^{np}.
\end{equation}
This series converges uniformly when \(x\) is restricted to any compact subset of \([0, \frac{1}{2})\).
We can therefore substitute \eqref{eq:renyi-entropy-series} into the integral \eqref{eq:renyi-fourier-localization} provided that we choose \(\chi(k)\) with sufficiently small support, e.g.,~such that \(\lambda_{-}(k) < \frac{1}{4}\) whenever \(\chi(k) \neq 0\).
Dropping the first term in the expansion, since it is smooth and therefore contributes only a superpolynomially small term, and writing \(\lambda_{-}(k) = k^{2} \xi(k)\), we use the evenness of \(\chi(k)\) and \(\lambda_{-}(k)\) to obtain
\begin{equation}\label{eq:renyi-fourier-integral-series}
    \hat{g}_{n}(m) = \frac{1}{\pi(1 - n)} \sum_{p = 1}^{\infty} \frac{(-1)^{p+1}}{p} \int_{0}^{\pi} \chi(k) \biggl(\frac{\xi(k)}{1 - \lambda_{-}(k)}\biggr)^{np} \abs{k}^{2np} \cos(m k) \dd{k} + o(\abs{m}^{-\infty}).
\end{equation}

The above integrals are of the form
\begin{equation}
    I_{m,t}[f] \defeq \int_{0}^{\pi} \chi(k) f(k) \abs{k}^{t} \cos(m k) \dd{k},
\end{equation}
where \(f(k)\) is a smooth even function with \(f(0) \neq 0\), and \(t > 0\).
These integrals are of the form considered in Ref.~\cite{Erdelyi:AsymptoticExpansions:1956} (Theorem~1 for \(t \in \mathbb{Z}\) and Theorem~2 for \(t \notin \mathbb{Z}\)).
Their leading asymptotic behavior is
\begin{equation}
    I_{m,t}[f] = -\Gamma(t + 1) \sin\Bigl(\frac{\pi t}{2}\Bigr) \frac{f(0)}{\abs{m}^{t+1}} + O(\abs{m}^{-t-3}).
\end{equation}

For noninteger \(n\), the \(p = 1\) term in \eqref{eq:renyi-fourier-integral-series} is the least regular and therefore determines the leading asymptotic behavior.
Since \(\xi(0) = \gamma\), we find
\begin{equation}\label{eq:renyi-fourier-asymptotics}
    \hat{g}_{n}(m) = -\frac{\Gamma(2n + 1) \sin(\pi n) \gamma^{n}}{\pi(1 - n)} \frac{1}{\abs{m}^{2n+1}} + o(\abs{m}^{-2n-1}).
\end{equation}
We conclude that the Fourier coefficients \(\hat{g}_{n}(m)\) decay \emph{algebraically} with \(\abs{m}\) when the R\'enyi parameter is not an integer.

Using the above result \eqref{eq:renyi-fourier-asymptotics}, we can now work out the system-size dependence of the order-\(n\) R\'enyi entropy in a ring with \(N\) unit cells.
To do this, we substitute \eqref{eq:renyi-fourier-asymptotics} into \eqref{eq:ring-entropy-fourier-sum}.
Summing over \(j \neq 0\) gives
\begin{equation}
    S_{n}(\rho_{\mathrm{ring}}) = N \bar{s}_{n} - \frac{2\Gamma(2n + 1) \sin(\pi n) \zeta(2n + 1) \gamma^{n}}{\pi(1 - n)} \frac{1}{N^{2n}} + o(N^{-2n}),
\end{equation}
where \(\bar{s}_{n} \defeq \hat{g}_{n}(0)\).
Using Riemann's functional equation, \(\zeta(s) = 2^{s} \pi^{s-1} \sin(\frac{\pi s}{2}) \Gamma(1 - s) \zeta(1 - s)\), with \(s = -2n\), this may be written more compactly as
\begin{equation}
    S_{n}(\rho_{\mathrm{ring}}) = N \bar{s}_{n} - \frac{2\zeta(-2n)}{n - 1} \frac{(4\pi^{2}\gamma)^{n}}{N^{2n}} + o(N^{-2n}).
\end{equation}
This calculation provides two useful consistency checks.
First, taking the limit \(n \to 1\) and using \(\zeta'(-2) = -\frac{\zeta(3)}{4\pi^{2}}\) recovers the von Neumann result \eqref{eq:ring-entropy-asymptotics}.
Also, for every integer \(n \geq 2\), \(\zeta(-2n) = 0\), consistent with the analyticity of \(h_{n}(x)\), which requires that the algebraic term vanish and the leading subextensive term be exponentially small.

\section{Alternative parametrization of eigenvalue bands}

In this appendix, we describe an alternative way to parametrize the two eigenvalue bands \eqref{eq:correlation-eigenvalue-bands} which we will find convenient for some of our calculations.

To explain this parametrization, we first write out the formula \(\lambda_{\pm}(k) = a(k) \pm \abs{b(k)}\) as an explicit function of \(k\):
\begin{equation}\label{eq:explicit-eigenvalue-bands}
    \lambda_{\pm}(k) = \frac{1}{2} \sin^{2}(\epsilon) \Biggl(1 - \frac{1}{2} \sin^{2}(\epsilon) \cos^{2}\Bigl(\frac{k}{2}\Bigr) \pm \abs[\Big]{\cos\Bigl(\frac{k}{2}\Bigr)} \sqrt{\cos^{2}(\epsilon) + \frac{1}{4} \sin^{4}(\epsilon) \cos^{2}\Bigl(\frac{k}{2}\Bigr)}\Biggr).
\end{equation}
Notice the absolute value, which introduces cusps into the individual functions \(\lambda_{\pm}(k)\) whenever the two eigenvalues coincide, even though the unordered pair \(\{\lambda_{-}(k), \lambda_{+}(k)\}\) admits an analytic parametrization across these values of \(k\).
This obscures the regularity of symmetric functions of the spectrum, such as \(g(k) = h(\lambda_{-}(k)) + h(\lambda_{+}(k))\), where the cusps in each summand cancel.

It is therefore more convenient to parametrize the spectrum using a single analytic branch.
Define
\begin{equation}\label{eq:analytic-eigenbranch}
    \lambda(k) \defeq \frac{1}{2} \sin^{2}(\epsilon) \Biggl(1 - \frac{1}{2} \sin^{2}(\epsilon) \cos^{2}\Bigl(\frac{k}{2}\Bigr) - \cos\Bigl(\frac{k}{2}\Bigr) \sqrt{\cos^{2}(\epsilon) + \frac{1}{4} \sin^{4}(\epsilon) \cos^{2}\Bigl(\frac{k}{2}\Bigr)}\Biggr).
\end{equation}
For each \(k\), \(\lambda(k)\) coincides with one of the two ordered eigenvalues \(\lambda_{-}(k)\) and \(\lambda_{+}(k)\), while \(\lambda(k + 2\pi)\) coincides with the other.
Thus, \(\lambda(k)\) is \(4\pi\)-periodic and \(\{\lambda_{-}(k), \lambda_{+}(k)\} = \{\lambda(k), \lambda(k + 2\pi)\}\).
Consequently,
\begin{equation}\label{eq:entropy-analytic-eigenbranch-form}
    g(k) = h\bigl(\lambda(k)\bigr) + h\bigl(\lambda(k + 2\pi)\bigr),
\end{equation}
which makes its analyticity at values of \(k\) where the two eigenvalues coincide manifest.

\section{Eigenvalues of \(C_{\mathrm{int}}\)}\label{app:interval-eigenvalues}

In this appendix, we derive the formula \eqref{eq:interval-low-lying-spectrum} for the low-lying eigenvalues of the correlation matrix \(C_{\mathrm{int}}\).
We present two independent derivations: (1) a physical argument based on general properties of reflection coefficients, and (2) a direct computation.

\subsection{Physical derivation of low-lying spectrum}

Consider an interval consisting of \(N\) unit cells, or equivalently, \(2N\) sites.
Let \(C_{\mathrm{int}}\) denote the correlation matrix \(C_{ij} = \expval{c_{i}^{\dagger} c_{j}}_{\rho}\), restricted to this interval.
Our goal is to derive the approximate low-lying spectrum of \(C_{\mathrm{int}}\) in the limit of large \(N\).

Before tackling the interval case, we examine the low-lying spectrum of \(C\) in some simpler geometries.
We begin with the low-lying eigenvectors of the unrestricted infinite-chain correlation matrix \(C\).
These eigenvectors take the form
\begin{equation}\label{eq:infinite-chain-eigenvectors}
    v_{2x+\mu} = v_{\mu}(k) e^{-i k x},
\end{equation}
where \(\mu \in \{0, 1\}\) indexes the site within each unit cell and \(v(k)\) is an eigenvector of the \(2 \times 2\) matrix \(C(k)\) defined in \eqref{eq:correlation-symbol}, with eigenvalue \(\lambda_{-}(k)\).
We will find it convenient to use the following convention for the phase and normalization of \(v(k)\):
\begin{equation}
    v(k) = \frac{1}{\sqrt{2 \abs{b(k)}}} \begin{pmatrix}
        \sqrt{b(k)} \\[1ex]
        -\sqrt{b^{*}(k)}
    \end{pmatrix},
\end{equation}
where \(b(k)\) is defined in \eqref{eq:correlation-symbol-offdiagonal}.
This convention has the useful properties that \(v(k)\) is normalized, is smooth near \(k = 0\), and satisfies \(\sigma_{x} v(k) = -v(-k)\).
The last property will be important below when we make use of reflection symmetry.

Next, consider the low-lying eigenvectors of \(C_{\text{half-chain}}\), the correlation matrix restricted to the \emph{half}-infinite chain of sites \(i \in \mathbb{Z}_{\geq 0}\).
We claim that these eigenvectors are all \emph{scattering} states rather than bound states.
Bound states have a discrete spectrum, and since \(C_{\text{half-chain}}\) is manifestly positive semidefinite, a bound state could contribute to the low-lying spectrum only if its eigenvalue were exactly \(0\).
This possibility can be ruled out using the same argument as in Lemma~\ref{lem:interval-eigenvalue-bounds} below.
We can therefore restrict our attention to scattering states.
Far from the boundary, these states approach superpositions of left- and right-moving eigenstates of the infinite chain given in \eqref{eq:infinite-chain-eigenvectors}:
\begin{equation}\label{eq:half-chain-scattering-state}
    v^{\text{half-chain}}_{2x+\mu} \sim v_{\mu}(k) e^{-i k x} + R(k) v_{\mu}(-k) e^{i k x},
\end{equation}
where \(R(k) \in \mathbb{C}\) and the limit \(x \to \infty\) is understood.
The corresponding eigenvalue is \(\lambda_{-}(\abs{k})\).

We can interpret \(R(k)\) as a reflection coefficient by considering a half-infinite chain of masses and springs with the Hamiltonian \(H = \frac{1}{2} \sum_{i = 0}^{\infty} p_{i}^{2} + \frac{1}{2} \sum_{i,j = 0}^{\infty} C_{ij} x_{i} x_{j}\).
The normal modes of this phonon system are in direct correspondence with the eigenvectors of \(C_{\text{half-chain}}\).
Likewise, \(R(k)\) is precisely the reflection coefficient associated with the end of the chain.

The reflection coefficient \(R(k)\) has two general properties that we will need below:
\begin{align}
    \label{eq:reflection-unit-modulus} \abs{R(k)} &= 1, \\
    \label{eq:reflection-reciprocity} R(k) R(-k) &= 1.
\end{align}
The first property, \eqref{eq:reflection-unit-modulus}, is easy to understand from the phonon analog discussed above: conservation of energy requires the reflected wave to carry the same energy flux as the incoming wave, so \(R(k)\) must have unit modulus.
For the second property, \eqref{eq:reflection-reciprocity}, recall that \(R(k)\) is the asymptotic ratio of the coefficients of \(v_{\mu}(-k) e^{i k x}\) and \(v_{\mu}(k) e^{-i k x}\) in the unique scattering state of \(C_{\text{half-chain}}\) with eigenvalue \(\lambda_{-}(\abs{k})\).
Replacing \(k\) by \(-k\) interchanges these coefficients and therefore inverts their ratio.

As a final warm-up, we consider the low-lying spectrum of \(C_{\text{ref-half-chain}}\), the correlation matrix restricted to a \emph{reflected} half-infinite chain, consisting of sites \(i \in \mathbb{Z}_{\leq -1}\).
This half-infinite chain is related to \(C_{\text{half-chain}}\) by reflection symmetry, so we can obtain its scattering states by applying the reflection map \(i \to -i - 1\) to \eqref{eq:half-chain-scattering-state}.
Using the fact that this map acts on the index \(\mu\) like \(\sigma_{x}\) and sends the unit-cell index \(x \to -x - 1\), we obtain
\begin{equation}\label{eq:reflected-half-chain-scattering-state}
    v^{\text{ref-half-chain}}_{2x+\mu} \sim v_{\mu}(-k) e^{i k (x+1)} + R(k) v_{\mu}(k) e^{-i k (x+1)}.
\end{equation}
Here we are using the fact that \(\sigma_{x} v(k) = -v(-k)\), as mentioned earlier, and we are dropping an overall factor of \(-1\).

We are now ready to consider the problem of interest, namely the low-lying eigenvectors of \(C_{\mathrm{int}}\) for a large interval.
Specifically, we consider an interval consisting of \(2N\) sites \(\{0, 1, \ldots, 2N - 1\}\) (i.e.,~\(N\) unit cells).
Like the half-infinite chain discussed above, we expect that all low-lying eigenvectors must approach a linear combination of left- and right-moving states \eqref{eq:infinite-chain-eigenvectors} in the interior of the interval:
\begin{equation}\label{eq:interval-scattering-ansatz}
    v^{\mathrm{int}}_{2x+\mu} \approx A v_{\mu}(k) e^{-i k x} + B v_{\mu}(-k) e^{i k x},
\end{equation}
where \(A, B \in \mathbb{C}\).
Furthermore, by \eqref{eq:half-chain-scattering-state} above, the boundary condition at the left end of the interval requires that
\begin{equation}\label{eq:interval-left-boundary-condition}
    \frac{B}{A} \approx R(k)
\end{equation}
in the limit of a large interval.
As for the boundary condition on the right end of the interval, we can read this off from \eqref{eq:reflected-half-chain-scattering-state} if we replace \(x \to x - N\) in \eqref{eq:reflected-half-chain-scattering-state} to take into account the fact that the right endpoint is at \(2N - 1\) rather than \(-1\).
In this way, we obtain the constraint
\begin{equation}\label{eq:interval-right-boundary-condition}
    \frac{A e^{-i k (N-1)}}{B e^{i k (N-1)}} = R(k).
\end{equation}

Combining \eqref{eq:interval-left-boundary-condition} and \eqref{eq:interval-right-boundary-condition} gives
\begin{equation}\label{eq:reflection-quantization}
    R(k)^{2} e^{2i k (N-1)} = 1.
\end{equation}
To determine the allowed values of \(k\), note that \eqref{eq:reflection-reciprocity} implies \(R(0)^{2} = 1\).
Assuming that \(R(k)\) is smooth, it follows that
\begin{equation}\label{eq:reflection-small-momentum}
    R(k)^{2} = 1 + O(k).
\end{equation}
Substituting \eqref{eq:reflection-small-momentum} into \eqref{eq:reflection-quantization}, we deduce that the allowed values of \(k\) are of the form
\begin{equation}\label{eq:low-lying-momentum-asymptotics}
    k = \frac{\pi j}{N} + O\Bigl(\frac{j}{N^{2}}\Bigr), \qquad j \in \{1, 2, \ldots\}.
\end{equation}
The low-lying eigenvalues are therefore of the form
\begin{equation}\label{eq:interval-low-lying-spectrum-error}
    \lambda_{-}\Bigl(\frac{\pi j}{N}\Bigr) + O\Bigl(\frac{j^{2}}{N^{3}}\Bigr), \qquad j \in \{1, 2, \ldots\},
\end{equation}
as we wished to show.

Strictly speaking, the analysis above applies only for \(k \neq 0\), since it assumes that the left- and right-moving modes in \eqref{eq:half-chain-scattering-state} are linearly independent.
We therefore need to determine separately whether \(C_{\mathrm{int}}\) has an eigenvector corresponding to \(k = 0\), or equivalently, an eigenvalue \(\lambda = 0\).
To answer this question, we rely on Lemma~\ref{lem:interval-eigenvalue-bounds} below, which shows that all eigenvalues of \(C_{\mathrm{int}}\) are \emph{strictly} larger than \(0\).
This rules out the possibility of an eigenvector with \(k = 0\) and explains why we do not include \(j = 0\) in \eqref{eq:low-lying-momentum-asymptotics} or \eqref{eq:interval-low-lying-spectrum-error}.

\subsection{Direct calculation of eigenvalues}\label{sec:direct-eigenvalue-calculation}

We now present a second derivation of the formula \eqref{eq:interval-low-lying-spectrum} for the low-lying eigenvalues of \(C_{\mathrm{int}}\).
In this derivation we will go a step further than \eqref{eq:interval-low-lying-spectrum} and characterize the \emph{full} eigenvalue spectrum of \(C_{\mathrm{int}}\).
While this full spectrum is not necessary for the argument given in the main text, we will use it in Appendix~\ref{app:interval-entropy-asymptotics}, in our rigorous derivation of the asymptotic formula \eqref{eq:interval-entropy-asymptotics} for the von Neumann entropy of an interval.

Our results are summarized by the following theorem (cf. Theorem~2.1 of Ref.~\cite{Bogoya:EigenvaluesHermitian:2015}), which characterizes the eigenvalues of the restricted correlation matrix \(C_{\mathrm{int}}\) for an interval of \(N\) unit cells.

\begin{theorem}\label{thm:interval-spectrum}
    For \(N\) sufficiently large, the eigenvalues of \(C_{\mathrm{int}}\) are all distinct and can be written as \(\sigma_{j} = \lambda(k_{j})\) for \({j \in \{1, 2, \ldots, 2N\}}\), where \(\lambda(k)\) is the analytic eigenbranch defined in \eqref{eq:analytic-eigenbranch} and \(k_{j} \in (0, 2\pi)\) is the unique solution of
    \begin{equation}\label{eq:interval-quantization}
        (N + 1) k_{j} + \eta(k_{j}) = \pi j.
    \end{equation}
    Here \(\eta(k)\) is an odd real-analytic function satisfying \(\eta(k + 4\pi) = \eta(k) - 2\pi\).
\end{theorem}

To derive \eqref{eq:interval-low-lying-spectrum} from this theorem, note that \(\eta'(k)\) is \(4\pi\)-periodic and therefore bounded.
Since \(\eta(0) = 0\), it follows that \(\abs{\eta(k)} \leq (\text{const.}) \abs{k}\).
Substituting this bound into \eqref{eq:interval-quantization} gives \(k_{j} = \frac{\pi j}{N} + O(\frac{j}{N^{2}})\), establishing \eqref{eq:interval-low-lying-spectrum}.

\begin{proof}
    We establish Theorem~\ref{thm:interval-spectrum} by proving four claims:
    \begin{enumerate}[label=(\roman*),ref=\roman*]
        \item \label{item:interval-spectrum-simplicity} The eigenvalues of \(C_{\mathrm{int}}\) are all distinct; equivalently, every eigenvalue has multiplicity one.

        \item \label{item:interval-spectrum-parametrization} Every eigenvalue \(\sigma\) of \(C_{\mathrm{int}}\) can be written as \(\sigma = \lambda(k)\) for a unique \(k \in (0, 2\pi)\).

        \item \label{item:interval-spectrum-quantization} There exists an odd real-analytic function \(\eta(k)\), satisfying \(\eta(k + 4\pi) = \eta(k) - 2\pi\), such that, for \(k \in (0, 2\pi)\), the value \(\lambda(k)\) is an eigenvalue of \(C_{\mathrm{int}}\) if and only if
        \begin{equation}\label{eq:interval-quantization-proof}
            (N + 1) k + \eta(k) = \pi j
        \end{equation}
        for some \(j \in \mathbb{Z}\).

        \item \label{item:interval-spectrum-root-count} For \(N\) sufficiently large, \eqref{eq:interval-quantization-proof} has a unique solution \(k_{j}\) for each \(j \in \mathbb{Z}\).
        Moreover, \(k_{j} \in (0, 2\pi)\) if and only if \(j \in \{1, 2, \ldots, 2N\}\).
    \end{enumerate}

    We begin with some notation.
    Recall that \(C_{\mathrm{int}}\) is a \(2N \times 2N\) matrix.
    The partition of the interval into \(N\) two-site unit cells naturally organizes it into an \(N \times N\) block matrix, with \(2 \times 2\) blocks given by
    \begin{equation}\label{eq:interval-correlation-blocks}
        (C_{\mathrm{int}})_{xy} = C_{x-y} = \frac{1}{2\pi} \int_{-\pi}^{\pi} C(k) e^{-i k (x-y)} \dd{k}, \qquad 1 \leq x, y \leq N.
    \end{equation}
    Here, \(C(k)\) is the \(2 \times 2\) momentum-space correlation matrix defined in \eqref{eq:correlation-symbol}.
    Likewise, we write a \(2N\)-component column vector as \(\Psi = (\psi_{1}, \psi_{2}, \ldots, \psi_{N})^{\mathsf{T}}\), where each \(\psi_{x} \in \mathbb{C}^{2}\) is a two-component vector.
    Finally, throughout the proof, we use \(\Arg\) to denote the principal branch of the complex argument on \(\mathbb{C} \setminus (-\infty, 0]\), taking values in \((-\pi, \pi)\).

    With this notation in place, we turn to the four claims~(\ref{item:interval-spectrum-simplicity})--(\ref{item:interval-spectrum-root-count}), beginning with (\ref{item:interval-spectrum-parametrization}).
    By Lemma~\ref{lem:interval-eigenvalue-bounds}, every eigenvalue \(\sigma\) of \(C_{\mathrm{int}}\) lies strictly between \(\lambda_{-}(0)\) and \(\lambda_{+}(0)\).
    At the same time, one can check that the analytic eigenbranch \(\lambda(k)\) defined in \eqref{eq:analytic-eigenbranch} is strictly increasing on \([0, 2\pi]\), with \(\lambda(0) = \lambda_{-}(0)\) and \(\lambda(2\pi) = \lambda_{+}(0)\).
    It follows that each eigenvalue \(\sigma\) can be written uniquely as \(\sigma = \lambda(k)\) for some \(k \in (0, 2\pi)\), establishing (\ref{item:interval-spectrum-parametrization}).

    To prove the other claims, we fix \(k \in (0, 2\pi)\) and define \(\Lambda \defeq C_{\mathrm{int}} - \lambda(k)\).
    Our task is to count the number of independent solutions to \(\Lambda \Psi = 0\).
    In components, this equation reads
    \begin{align}
        \label{eq:interval-eigenvalue-left-endpoint} \bigl(C_{0} - \lambda(k)\bigr) \psi_{1} + C_{-1} \psi_{2} + C_{-2} \psi_{3} &= 0 \qquad (x = 1), \\[1ex]
        \label{eq:interval-eigenvalue-left-interior} C_{1} \psi_{1} + \bigl(C_{0} - \lambda(k)\bigr) \psi_{2} + C_{-1} \psi_{3} + C_{-2} \psi_{4} &= 0 \qquad (x = 2), \\[1ex]
        \label{eq:interval-eigenvalue-bulk} C_{2} \psi_{x-2} + C_{1} \psi_{x-1} + \bigl(C_{0} - \lambda(k)\bigr) \psi_{x} + C_{-1} \psi_{x+1} + C_{-2} \psi_{x+2} &= 0 \qquad (3 \leq x \leq N - 2), \\[1ex]
        \label{eq:interval-eigenvalue-right-interior} C_{2} \psi_{N-3} + C_{1} \psi_{N-2} + \bigl(C_{0} - \lambda(k)\bigr) \psi_{N-1} + C_{-1} \psi_{N} &= 0 \qquad (x = N - 1), \\[1ex]
        \label{eq:interval-eigenvalue-right-endpoint} C_{2} \psi_{N-2} + C_{1} \psi_{N-1} + \bigl(C_{0} - \lambda(k)\bigr) \psi_{N} &= 0 \qquad (x = N).
    \end{align}
    For the moment, we set aside the two endpoint equations \eqref{eq:interval-eigenvalue-left-endpoint} and \eqref{eq:interval-eigenvalue-right-endpoint} and look for vectors satisfying only \eqref{eq:interval-eigenvalue-left-interior}--\eqref{eq:interval-eigenvalue-right-interior}.
    We call such vectors ``interior.''
    Our immediate goal is to construct an explicit basis for the subspace of interior vectors \(\mathcal{I} \subset \mathbb{C}^{2N}\).

    We start by considering \eqref{eq:interval-eigenvalue-bulk}.
    Away from the endpoints, this equation is identical to the eigenvalue equation for the infinite chain.
    Two independent plane-wave solutions with eigenvalue \(\lambda(k)\) are
    \begin{equation}\label{eq:interior-plane-waves}
        \psi_{x}^{(+)} = v_{+} e^{-i k x}, \qquad \psi_{x}^{(-)} = v_{-} e^{i k x},
    \end{equation}
    where \(v_{+}, v_{-} \in \mathbb{C}^{2}\) are fixed up to normalization by \(C(\pm k) v_{\pm} = \lambda(k) v_{\pm}\).
    We use the following phase convention for \(v_{\pm}\):
    \begin{equation}
        v_{\pm} = 2 \cot^{2}\Bigl(\frac{\epsilon}{2}\Bigr) \begin{pmatrix}
            1 \\[1ex]
            -e^{\mp i \varphi(k)}
        \end{pmatrix}, \qquad \varphi(k) \defeq \frac{k}{2} + \Arg\Bigl(\cos^{4}\Bigl(\frac{\epsilon}{2}\Bigr) + \sin^{4}\Bigl(\frac{\epsilon}{2}\Bigr) e^{i k}\Bigr).
    \end{equation}
    The phase \(\varphi(k)\) is chosen so that \(b(k) = e^{i \varphi(k)} \bigl(a(k) - \lambda(k)\bigr)\).
    Note that \(\varphi(k)\) is odd and satisfies \(\varphi(k + 2\pi) = \varphi(k) + \pi\).

    Restricting these plane waves to \(1 \leq x \leq N\) automatically satisfies \eqref{eq:interval-eigenvalue-bulk}.
    Near the ends of the interval, however, \eqref{eq:interval-eigenvalue-left-interior} and \eqref{eq:interval-eigenvalue-right-interior} differ from the infinite-chain equation because the cells \(0\) and \(N + 1\) are absent.
    We compensate for these missing terms by adjusting the first and last components:
    \begin{equation}\label{eq:boundary-adjusted-plane-waves}
        \Psi^{(\pm)} = \bigl(\psi_{1}^{(\pm)} + \delta\psi_{1}^{(\pm)}, \psi_{2}^{(\pm)}, \ldots, \psi_{N-1}^{(\pm)}, \psi_{N}^{(\pm)} + \delta\psi_{N}^{(\pm)}\bigr)^{\mathsf{T}}.
    \end{equation}
    Substituting this ansatz into the interior equations gives
    \begin{align}
        -C_{2} \psi_{0}^{(\pm)} + C_{1} \delta\psi_{1}^{(\pm)} &= 0 \qquad (x = 2), \\[1ex]
        C_{2} \delta\psi_{1}^{(\pm)} &= 0 \qquad (x = 3), \\[1ex]
        C_{-2} \delta\psi_{N}^{(\pm)} &= 0 \qquad (x = N - 2), \\[1ex]
        C_{-1} \delta\psi_{N}^{(\pm)} - C_{-2} \psi_{N+1}^{(\pm)} &= 0 \qquad (x = N - 1).
    \end{align}
    Since \(C_{1}\) and \(C_{-1}\) are invertible, the first and last equations determine the endpoint adjustments:
    \begin{equation}
        \delta\psi_{1}^{(\pm)} = C_{1}^{-1} C_{2} \psi_{0}^{(\pm)}, \qquad \delta\psi_{N}^{(\pm)} = C_{-1}^{-1} C_{-2} \psi_{N+1}^{(\pm)}.
    \end{equation}
    The two middle equations are then automatically satisfied because \(C_{2} C_{1}^{-1} C_{2} = 0\) and \(C_{-2} C_{-1}^{-1} C_{-2} = 0\).
    Thus, the vectors \(\Psi^{(+)}\) and \(\Psi^{(-)}\) are both interior and provide two basis vectors for the subspace \(\mathcal{I}\).

    We next construct two more interior vectors that are localized near the endpoints of the interval.
    We take
    \begin{equation}
        \Psi^{(L)} = \bigl(\psi^{(L)}_{1}, \psi^{(L)}_{2}, 0, \ldots, 0\bigr)^{\mathsf{T}}, \qquad \Psi^{(R)} = \bigl(0, \ldots, 0, \psi^{(R)}_{N-1}, \psi^{(R)}_{N}\bigr)^{\mathsf{T}}.
    \end{equation}
    Substituting these vectors into the interior equations gives
    \begin{alignat}{2}
        C_{1} \psi^{(L)}_{1} + \bigl(C_{0} - \lambda(k)\bigr) \psi^{(L)}_{2} &= 0 \qquad (x = 2), &\qquad \bigl(C_{0} - \lambda(k)\bigr) \psi^{(R)}_{N-1} + C_{-1} \psi^{(R)}_{N} &= 0 \qquad (x = N - 1), \\[1ex]
        C_{2} \psi^{(L)}_{1} + C_{1} \psi^{(L)}_{2} &= 0 \qquad (x = 3), &\qquad C_{-1} \psi^{(R)}_{N-1} + C_{-2} \psi^{(R)}_{N} &= 0 \qquad (x = N - 2), \\[1ex]
        C_{2} \psi^{(L)}_{2} &= 0 \qquad (x = 4), &\qquad C_{-2} \psi^{(R)}_{N-1} &= 0 \qquad (x = N - 3).
    \end{alignat}
    The first equation in each column determines \(\psi^{(L)}_{1}\) and \(\psi^{(R)}_{N}\) in terms of \(\psi^{(L)}_{2}\) and \(\psi^{(R)}_{N-1}\), respectively.
    The last equation in each column requires \(\psi^{(L)}_{2}\) and \(\psi^{(R)}_{N-1}\) to lie in the kernels of \(C_{2}\) and \(C_{-2}\), respectively.
    A convenient choice is
    \begin{alignat}{2}
        \psi^{(L)}_{1} &= -C_{1}^{-1} \bigl(C_{0} - \lambda(k)\bigr) \psi^{(L)}_{2}, &\qquad \psi^{(R)}_{N} &= -C_{-1}^{-1} \bigl(C_{0} - \lambda(k)\bigr) \psi^{(R)}_{N-1}, \\[1ex]
        \psi^{(L)}_{2} &= \tan^{2}\Bigl(\frac{\epsilon}{2}\Bigr) \begin{pmatrix}
            1 \\[1ex]
            0
        \end{pmatrix}, &\qquad \psi^{(R)}_{N-1} &= \tan^{2}\Bigl(\frac{\epsilon}{2}\Bigr) \begin{pmatrix}
            0 \\[1ex]
            1
        \end{pmatrix}.
    \end{alignat}
    With these choices, direct substitution shows that the middle equation in each column is also satisfied.
    Thus, \(\Psi^{(L)}\) and \(\Psi^{(R)}\) are interior and provide two more basis vectors for \(\mathcal{I}\).

    We now invoke Lemma~\ref{lem:interior-solution-dimension}, which tells us that \(\mathcal{I}\) is in fact four-dimensional.
    It follows that \(\{\Psi^{(+)}, \Psi^{(-)}, \Psi^{(L)}, \Psi^{(R)}\}\) forms a \emph{complete} basis for \(\mathcal{I}\), so every interior vector can be written uniquely in the form \(\kappa_{1} \Psi^{(+)} + \kappa_{2} \Psi^{(-)} + \kappa_{3} \Psi^{(L)} + \kappa_{4} \Psi^{(R)}\).

    Our next task is to determine which of these linear combinations satisfy the two endpoint equations, \eqref{eq:interval-eigenvalue-left-endpoint} and \eqref{eq:interval-eigenvalue-right-endpoint}.
    Upon substituting the four basis vectors, the only nonzero terms in the endpoint equations are
    \begin{alignat}{2}
        \bigl(\Lambda \Psi^{(\pm)}\bigr)_{1} &= \begin{pmatrix}
            e^{\mp i \varphi(k)} \mathcal{R}(\pm k) \\[1ex]
            0
        \end{pmatrix}, &\qquad \bigl(\Lambda \Psi^{(\pm)}\bigr)_{N} &= \begin{pmatrix}
            0 \\[1ex]
            -e^{\mp i k (N+1)} \mathcal{R}(\mp k)
        \end{pmatrix}, \\[2ex]
        \bigl(\Lambda \Psi^{(L)}\bigr)_{1} &= \begin{pmatrix}
            \alpha(k) \\[1ex]
            1 - 2\lambda(k)
        \end{pmatrix}, &\qquad \bigl(\Lambda \Psi^{(R)}\bigr)_{N} &= \begin{pmatrix}
            1 - 2\lambda(k) \\[1ex]
            \alpha(k)
        \end{pmatrix}.
    \end{alignat}
    Here, \(\alpha(k)\) is a function whose explicit form will not be needed, while
    \begin{equation}\label{eq:boundary-residual-function}
        \mathcal{R}(k) \defeq 1 - 2\lambda(k) + \frac{1}{2} \cos(\epsilon) \bigl(1 + \cos^{2}(\epsilon)\bigr) + \frac{1}{8} \sin^{4}(\epsilon) e^{i k} + \frac{1}{2} \sin^{2}(\epsilon) \cos^{4}\Bigl(\frac{\epsilon}{2}\Bigr) e^{i \varphi(k)}.
    \end{equation}
    We note the conjugate symmetry \(\mathcal{R}(-k) = \overline{\mathcal{R}(k)}\).

    Using these expressions, the second component of \eqref{eq:interval-eigenvalue-left-endpoint} and the first component of \eqref{eq:interval-eigenvalue-right-endpoint} force \(\kappa_{3} = \kappa_{4} = 0\), since \(1 - 2\lambda(k) \neq 0\).
    The two remaining conditions are therefore
    \begin{equation}\label{eq:boundary-coefficient-system}
        \begin{pmatrix}
            e^{-i \varphi(k)} \mathcal{R}(k) & e^{i \varphi(k)} \mathcal{R}(-k) \\[2ex]
            -e^{-i k (N+1)} \mathcal{R}(-k) & -e^{i k (N+1)} \mathcal{R}(k)
        \end{pmatrix} \begin{pmatrix}
            \kappa_{1} \\[2ex]
            \kappa_{2}
        \end{pmatrix} = 0.
    \end{equation}
    A nonzero solution exists if and only if the determinant vanishes, which gives
    \begin{equation}\label{eq:boundary-determinant-condition}
        R(k)^{2} e^{2i k (N+1)} = 1, \qquad R(k) \defeq e^{-i \varphi(k)} \frac{\mathcal{R}(k)}{\mathcal{R}(-k)}.
    \end{equation}
    Lemma~\ref{lem:boundary-residual-branch-cut} shows that \(\mathcal{R}(k) \notin (-\infty, 0]\) for all \(k\), so we may define the analytic function \(\eta(k) \defeq -\varphi(k) + 2 \Arg(\mathcal{R}(k))\).
    The symmetry and periodicity properties of \(\varphi\) and \(\mathcal{R}\) imply that \(\eta\) is odd and satisfies \(\eta(k + 4\pi) = \eta(k) - 2\pi\).
    The above determinant condition \eqref{eq:boundary-determinant-condition} is then equivalent to \eqref{eq:interval-quantization-proof}, establishing (\ref{item:interval-spectrum-quantization}).

    To establish (\ref{item:interval-spectrum-simplicity}), suppose \eqref{eq:boundary-coefficient-system} had two linearly independent solutions.
    These would span all possible choices of \((\kappa_{1}, \kappa_{2})\), so the matrix in \eqref{eq:boundary-coefficient-system} would have to vanish, which is impossible because \(\mathcal{R}(k) \neq 0\).
    Hence, whenever \(\lambda(k)\) is an eigenvalue, it has multiplicity one.

    Finally, to establish (\ref{item:interval-spectrum-root-count}), we rewrite \eqref{eq:interval-quantization-proof} as
    \begin{equation}\label{eq:interval-quantization-map}
        \phi_{N}(k) = \frac{\pi j}{N + 1}, \qquad \phi_{N}(k) \defeq k + \frac{\eta(k)}{N + 1}.
    \end{equation}
    The symmetries of \(\eta\) imply that \(\eta'\) is periodic and bounded.
    Therefore, for \(N\) sufficiently large, \(\phi_{N}'(k) > \frac{1}{2}\) for all \(k\), and \(\phi_{N}\) is a bijection from \(\mathbb{R}\) to \(\mathbb{R}\).
    Equation~\eqref{eq:interval-quantization-map} thus has a unique solution \(k_{j}\) for every \(j \in \mathbb{Z}\).

    Moreover, \(k_{j} \in (0, 2\pi)\) if and only if \(\phi_{N}(0) < \frac{\pi j}{N + 1} < \phi_{N}(2\pi)\).
    Since \(\eta(2\pi) = -\pi\), this reduces to \(0 < j < 2N + 1\), or equivalently \(j \in \{1, 2, \ldots, 2N\}\), establishing (\ref{item:interval-spectrum-root-count}) and completing the proof.
\end{proof}

\subsection{Supporting lemmas}

\begin{lemma}\label{lem:interval-eigenvalue-bounds}
    Every eigenvalue \(\sigma\) of \(C_{\mathrm{int}}\) satisfies \(\lambda_{-}(0) < \sigma < \lambda_{+}(0)\).
\end{lemma}

\begin{proof}
    Let \(\Psi = (\psi_{1}, \ldots, \psi_{N})^{\mathsf{T}}\) be a normalized eigenvector of \(C_{\mathrm{int}}\) with eigenvalue \(\sigma\).
    Define \(\psi(k) \defeq \sum_{x = 1}^{N} \psi_{x} e^{i k x}\).
    By Parseval's identity, \(\frac{1}{2\pi} \int_{-\pi}^{\pi} \norm{\psi(k)}^{2} \dd{k} = 1\).
    Using the Fourier representation of the blocks \(C_{x-y}\), we obtain
    \begin{equation}
        \sigma = \Psi^{\dagger} C_{\mathrm{int}} \Psi = \sum_{x,y = 1}^{N} \psi_{x}^{\dagger} C_{x-y} \psi_{y} = \frac{1}{2\pi} \int_{-\pi}^{\pi} \psi^{\dagger}(k) C(k) \psi(k) \dd{k}.
    \end{equation}
    Since \(C(k)\) is Hermitian, \(\lambda_{-}(k) \norm{\psi(k)}^{2} \leq \psi^{\dagger}(k) C(k) \psi(k) \leq \lambda_{+}(k) \norm{\psi(k)}^{2}\).
    Moreover, the explicit expression \eqref{eq:explicit-eigenvalue-bands} for the eigenvalue bands shows that \(\lambda_{-}(k) \geq \lambda_{-}(0)\) and \(\lambda_{+}(k) \leq \lambda_{+}(0)\), with equality only at \(k = 0\).
    It follows that
    \begin{align}
        \sigma - \lambda_{-}(0) &\geq \frac{1}{2\pi} \int_{-\pi}^{\pi} \bigl(\lambda_{-}(k) - \lambda_{-}(0)\bigr) \norm{\psi(k)}^{2} \dd{k}, \\[1ex]
        \lambda_{+}(0) - \sigma &\geq \frac{1}{2\pi} \int_{-\pi}^{\pi} \bigl(\lambda_{+}(0) - \lambda_{+}(k)\bigr) \norm{\psi(k)}^{2} \dd{k}.
    \end{align}
    Since \(\Psi \neq 0\), the function \(\psi(k)\) is a nonzero vector-valued trigonometric polynomial, so \(\norm{\psi(k)}^{2}\) is positive almost everywhere.
    The factors in parentheses are strictly positive for \(k \neq 0\), and hence both integrals are strictly positive.
    Therefore, \(\lambda_{-}(0) < \sigma < \lambda_{+}(0)\).
\end{proof}

\begin{lemma}\label{lem:interior-solution-dimension}
    Let \(k \in (0, 2\pi)\) and define \(\Lambda \defeq C_{\mathrm{int}} - \lambda(k)\).
    The space
    \begin{equation}\label{eq:interior-solution-space}
        \mathcal{I} \defeq \Bigl\{\Psi = (\psi_{1}, \ldots, \psi_{N})^{\mathsf{T}} : (\Lambda \Psi)_{x} = 0 \text{ for } 2 \leq x \leq N - 1 \Bigr\}
    \end{equation}
    is four-dimensional.
\end{lemma}

\begin{proof}
    The equations \((\Lambda \Psi)_{x} = 0\) for \(2 \leq x \leq N - 1\) constitute a homogeneous linear system of \(2(N - 2)\) equations for the \(2N\) components of \(\Psi\).
    Hence, \(\dim \mathcal{I} \geq 4\).

    Now let \(V\) be the space of vectors of the form \(\Phi = (0, \phi_{2}, \ldots, \phi_{N-1}, 0)^{\mathsf{T}}\).
    For any \(\Phi \in V\) and \(\Psi \in \mathcal{I}\), the Hermiticity of \(\Lambda\) gives
    \begin{equation}
        (\Lambda \Phi)^{\dagger} \Psi = \Phi^{\dagger} \Lambda \Psi = \sum_{x = 2}^{N-1} \phi_{x}^{\dagger} (\Lambda \Psi)_{x} = 0.
    \end{equation}
    Thus, \(\Lambda \Phi\) is orthogonal to \(\mathcal{I}\), so \(\Lambda\) maps \(V\) into \(\mathcal{I}^{\perp}\).

    If \(\dim \mathcal{I} > 4\), then
    \begin{equation}
        \dim \mathcal{I}^{\perp} = 2N - \dim \mathcal{I}  < 2N - 4 = \dim V.
    \end{equation}
    The map from \(V\) to \(\mathcal{I}^{\perp}\) therefore cannot be injective.
    Consequently, there exists a nonzero vector \(\Phi = (0, \phi_{2}, \ldots, \phi_{N-1}, 0)^{\mathsf{T}}\) such that \(\Lambda \Phi = 0\).

    Extend the interval by one cell on either side and make the periodic identification to obtain a ring of \(N + 2\) unit cells with correlation matrix \(C_{\mathrm{ring}}\).
    Taking \(\Phi\) to vanish on the two added cells, i.e.,
    \begin{equation}
        \Phi = (0, 0, \phi_{2}, \ldots, \phi_{N-1}, 0, 0)^{\mathsf{T}},
    \end{equation}
    we claim that \(\Phi\) is an eigenvector of \(C_{\mathrm{ring}}\) with eigenvalue \(\lambda(k)\).
    The eigenvalue equation on the ring is
    \begin{equation}
        C_{2} \phi_{x-2} + C_{1} \phi_{x-1} + \bigl(C_{0} - \lambda(k)\bigr) \phi_{x} + C_{-1} \phi_{x+1} + C_{-2} \phi_{x+2} = 0 \qquad (0 \leq x \leq N + 1),
    \end{equation}
    where the components are understood periodically, \(\phi_{x+(N+2)} \defeq \phi_{x}\).
    For \(1 \leq x \leq N\), these equations are precisely the components of \(\Lambda \Phi = 0\), so it remains only to check
    \begin{align}
        \label{eq:ring-extension-left-condition} C_{-2} \phi_{2} &= 0 \qquad (x = 0), \\
        \label{eq:ring-extension-right-condition} C_{2} \phi_{N-1} &= 0 \qquad (x = N + 1).
    \end{align}
    Since \(C_{-1}\) and \(C_{1}\) are invertible, the equations at \(x = 1\) and \(x = N\) give
    \begin{equation}
        \phi_{2} = -C_{-1}^{-1} C_{-2} \phi_{3}, \qquad \phi_{N-1} = -C_{1}^{-1} C_{2} \phi_{N-2}.
    \end{equation}
    The identities \(C_{-2} C_{-1}^{-1} C_{-2} = 0\) and \(C_{2} C_{1}^{-1} C_{2} = 0\) then imply \eqref{eq:ring-extension-left-condition} and \eqref{eq:ring-extension-right-condition}.

    Hence, \(\Phi\) is an eigenvector of \(C_{\mathrm{ring}}\) with eigenvalue \(\lambda(k)\).
    Since \(C_{\mathrm{ring}}\) is translation invariant, the same is true of every translate of \(\Phi\).
    In particular, \(\Phi\) and its first two successive translates are linearly independent: if \(\phi_{r}\) is the first nonzero block of \(\Phi\), then the first nonzero blocks of these three vectors occur at \(r\), \(r + 1\), and \(r + 2\), respectively.
    Thus, \(\lambda(k)\) has multiplicity at least three.

    On the other hand, the exact diagonalization of \(C_{\mathrm{ring}}\) in \(k\)-space gives its \(2N + 4\) eigenvalues as
    \begin{equation}
        \sigma_{j} = \lambda\Bigl(\frac{2\pi j}{N + 2}\Bigr) \qquad (0 \leq j \leq 2N + 3).
    \end{equation}
    Since any value of \(\lambda\) occurs at most twice on \([0, 4\pi)\), every eigenvalue of \(C_{\mathrm{ring}}\) has multiplicity at most two, giving a contradiction.
    Therefore, \(\dim \mathcal{I} = 4\).
\end{proof}

\begin{lemma}\label{lem:boundary-residual-branch-cut}
    The function \(\mathcal{R}(k)\) in \eqref{eq:boundary-residual-function} satisfies \(\mathcal{R}(k) \notin (-\infty, 0]\) for all \(k \in \mathbb{R}\).
\end{lemma}

\begin{proof}
    It suffices to show that \(\mathcal{R}(k) > 0\) whenever \(\Im(\mathcal{R}(k)) = 0\).
    Since \(\mathcal{R}(k)\) is \(4\pi\)-periodic, we may restrict to \(k \in [0, 4\pi)\).

    Put \(\tau \defeq \tan^{4}(\frac{\epsilon}{2})\) and \(D(k) \defeq \sqrt{(1 - \tau)^{2} + 4\tau \cos^{2}(\frac{k}{2})}\).
    A direct calculation gives
    \begin{equation}\label{eq:boundary-residual-imaginary-part}
        \Im\bigl(\mathcal{R}(k)\bigr) = \frac{1}{2} \sin^{2}(\epsilon) \cos^{4}\Bigl(\frac{\epsilon}{2}\Bigr) \sin\Bigl(\frac{k}{2}\Bigr) \biggl(\frac{\tau(1 - \tau)}{D(k)} + D(k) + 2\sqrt{\tau} \cos\Bigl(\frac{k}{2}\Bigr)\!\biggr).
    \end{equation}
    Since \(0 < \tau < 1\), we have that \(D(k) > 2\sqrt{\tau} \abs{\cos(\frac{k}{2})}\).
    The expression in parentheses is therefore strictly positive for every \(k\).
    Since the remaining prefactors are also nonzero, \(\Im(\mathcal{R}(k)) = 0\) if and only if \(\sin(\frac{k}{2}) = 0\).
    Within \([0, 4\pi)\), this occurs only at \(k = 0\) and \(k = 2\pi\).

    Evaluating \(\mathcal{R}(k)\) at these two points gives
    \begin{align}
        \mathcal{R}(0) &= 2 \cos^{2}\Bigl(\frac{\epsilon}{2}\Bigr) \biggl(1 + \sin^{4}\Bigl(\frac{\epsilon}{2}\Bigr)\!\biggr), \\
        \mathcal{R}(2\pi) &= 2 \cos(\epsilon) \cos^{2}\Bigl(\frac{\epsilon}{2}\Bigr) \biggl(\cos^{2}(\epsilon) + \sin^{4}\Bigl(\frac{\epsilon}{2}\Bigr)\!\biggr).
    \end{align}
    Both expressions are strictly positive.
    Thus, every point at which \(\mathcal{R}(k)\) is real lies on the positive real axis.
\end{proof}

\section{Asymptotics for interval entropies}\label{app:interval-entropy-asymptotics}

In this appendix, we present a rigorous derivation of the asymptotic formula \eqref{eq:interval-entropy-asymptotics} for the von Neumann entropy of an interval.
In contrast to the heuristic derivation given in the main text, the analysis below determines all three terms in \eqref{eq:interval-entropy-asymptotics}, including the constant term.
We also derive the analogous asymptotic formula \eqref{eq:interval-renyi-entropy-asymptotics} for the order-\(n\) R\'enyi entropy.

\subsection{Rigorous derivation of von Neumann entropy \(S(\rho_{\mathrm{int}})\)}

By \eqref{eq:gaussian-entropy} from the main text, the von Neumann entropy of an interval with \(N\) unit cells is given by
\begin{equation}\label{eq:interval-entropy-eigenvalue-sum}
    S(\rho_{\mathrm{int}}) = \Tr h(C_{\mathrm{int}}) = \sum_{j = 1}^{2N} h(\sigma_{j}), \qquad h(x) \defeq -x \log(x) - (1 - x) \log(1 - x),
\end{equation}
where \(\{\sigma_{j}\}_{j = 1}^{2N}\) are the eigenvalues of the restricted correlation matrix \(C_{\mathrm{int}}\).
Our task is to determine the asymptotic behavior of this sum as \(N \to \infty\).
We do this using the following result, proven in Appendix~\ref{sec:direct-eigenvalue-calculation}.

\begin{restatedtheorem}
    For \(N\) sufficiently large, the eigenvalues of \(C_{\mathrm{int}}\) are all distinct and can be written as \(\sigma_{j} = \lambda(k_{j})\) for \(j \in \{1, 2, \ldots, 2N\}\), where \(\lambda(k)\) is the analytic eigenbranch defined in \eqref{eq:analytic-eigenbranch} and \(k_{j} \in (0, 2\pi)\) is the unique solution of
    \begin{equation}\label{eq:interval-quantization-restated}
        (N + 1) k_{j} + \eta(k_{j}) = \pi j.
    \end{equation}
    Here \(\eta(k)\) is an odd real-analytic function satisfying \(\eta(k + 4\pi) = \eta(k) - 2\pi\).
\end{restatedtheorem}

First, we rewrite the quantization condition \eqref{eq:interval-quantization-restated} in a more explicit form.
Define \(\phi_{N}(k) \defeq k + \frac{\eta(k)}{N + 1}\).
We claim that \(\phi_{N}\) is strictly increasing for large \(N\).
To see this, notice that the derivative \(\eta'(k)\) is \(4\pi\)-periodic and therefore bounded.
Hence \(\phi_{N}'(k) = 1 + \frac{\eta'(k)}{N + 1} > 0\) for sufficiently large \(N\).

Since \(\phi_{N}\) is strictly increasing, it is invertible.
Set \(q_{N} \defeq \phi_{N}^{-1}\).
The quantization condition \eqref{eq:interval-quantization-restated} can then be rewritten as \(k_{j} = q_{N}(\frac{\pi j}{N + 1})\).
Substituting this expression into the entropy formula \eqref{eq:interval-entropy-eigenvalue-sum}, we obtain
\begin{equation}\label{eq:interval-entropy-rescaled-sum}
    S(\rho_{\mathrm{int}}) = \sum_{j = 1}^{2N} h\biggl(\lambda_{N}\Bigl(\frac{j}{2N + 1}\Bigr)\biggr), \qquad \lambda_{N}(u) \defeq \lambda\biggl(q_{N}\Bigl(\frac{\pi(2N + 1)}{N + 1} u\Bigr)\biggr).
\end{equation}

Next, we evaluate the sum \eqref{eq:interval-entropy-rescaled-sum}.
To do this, we need a few properties of \(\lambda_{N}\), which are summarized by the following lemma.

\begin{lemma}\label{lem:rescaled-eigenbranch-properties}~
    \begin{enumerate}[label=(\roman*),ref=\roman*]
        \item \label{item:rescaled-eigenbranch-endpoints} \(\lambda_{N}(0) = 0\) and \(\lambda_{N}(1) = \lambda(2\pi)\).
        \item \label{item:rescaled-eigenbranch-symmetry} \(\lambda_{N}(u)\) is even about both \(u = 0\) and \(u = 1\); that is, \(\lambda_{N}(-u) = \lambda_{N}(u)\) and \(\lambda_{N}(1 - u) = \lambda_{N}(1 + u)\).
        \item \label{item:rescaled-eigenbranch-uniform-bounds} \(\lambda_{N}\) and its derivatives are uniformly bounded in \(N\).
        More precisely, for every nonnegative integer \(r\), there exists a constant \(\Lambda_{r} > 0\), independent of \(N\), such that \(\abs{\lambda_{N}^{(r)}(u)} \leq \Lambda_{r}\) for all \(u \in \mathbb{R}\) and all sufficiently large \(N\).
    \end{enumerate}
\end{lemma}

\begin{proof}
    Set \(\kappa_{N} \defeq \frac{\pi(2N + 1)}{N + 1}\), so that \(\lambda_{N}(u) = \lambda(q_{N}(\kappa_{N} u))\).
    Since \(\phi_{N}\) is odd, \(q_{N}\) is also odd.
    Moreover, the property \(\eta(k + 4\pi) = \eta(k) - 2\pi\) implies \(q_{N}(x + 2\kappa_{N}) = q_{N}(x) + 4\pi\).
    Combining these two properties gives
    \begin{equation}\label{eq:inverse-quantization-reflection}
        q_{N}(2\kappa_{N} - x) = 4\pi - q_{N}(x).
    \end{equation}

    The oddness of \(q_{N}\) gives \(q_{N}(0) = 0\), while setting \(x = \kappa_{N}\) in \eqref{eq:inverse-quantization-reflection} gives \(q_{N}(\kappa_{N}) = 2\pi\).
    Therefore, \(\lambda_{N}(0) = \lambda(0) = 0\) and \(\lambda_{N}(1) = \lambda(2\pi)\), proving (\ref{item:rescaled-eigenbranch-endpoints}).

    Since \(\lambda(k)\) is even, the oddness of \(q_{N}\) implies \(\lambda_{N}(-u) = \lambda_{N}(u)\).
    Similarly, \eqref{eq:inverse-quantization-reflection} gives \(q_{N}(\kappa_{N}(1 - u)) = 4\pi - q_{N}(\kappa_{N}(1 + u))\).
    The evenness and \(4\pi\)-periodicity of \(\lambda(k)\) then imply \(\lambda_{N}(1 - u) = \lambda_{N}(1 + u)\), proving (\ref{item:rescaled-eigenbranch-symmetry}).

    To prove (\ref{item:rescaled-eigenbranch-uniform-bounds}), note that all positive-order derivatives of \(\phi_{N}\) are bounded uniformly in \(N\), while \(\phi_{N}'\) is uniformly bounded away from zero.
    Differentiating the identity \(\phi_{N}\bigl(q_{N}(x)\bigr) = x\) gives
    \begin{equation}
        q_{N}'(x) = \frac{1}{\phi_{N}'\bigl(q_{N}(x)\bigr)},
    \end{equation}
    so \(q_{N}'\) is uniformly bounded.
    Repeated differentiation of the same identity expresses each higher derivative of \(q_{N}\) in terms of bounded derivatives of \(\phi_{N}\), lower-order derivatives of \(q_{N}\), and inverse powers of \(\phi_{N}'(q_{N}(x))\).
    It follows inductively that all derivatives of \(q_{N}\) are uniformly bounded in \(N\).
    Finally, \(\kappa_{N} < 2\pi\), while \(\lambda\) and all of its derivatives are bounded because \(\lambda\) is smooth and \(4\pi\)-periodic.
    Repeated application of the chain rule to \(\lambda_{N}(u) = \lambda(q_{N}(\kappa_{N} u))\) yields the desired \(N\)-independent bounds.
\end{proof}

We are now ready to evaluate \eqref{eq:interval-entropy-rescaled-sum}.
Since \(\lambda_{N}(1) = \lambda(2\pi)\), we first extend the sum to include \(j = 2N + 1\):
\begin{equation}
    S(\rho_{\mathrm{int}}) = \sum_{j = 1}^{2N+1} h\biggl(\lambda_{N}\Bigl(\frac{j}{2N + 1}\Bigr)\biggr) - h\bigl(\lambda(2\pi)\bigr).
\end{equation}
The resulting sum is naturally analyzed using the Euler--Maclaurin formula:
\begin{equation}\label{eq:interval-entropy-euler-maclaurin}
    S(\rho_{\mathrm{int}}) = (2N + 1) \int_{0}^{1} h\bigl(\lambda_{N}(u)\bigr) \dd{u} - h\bigl(\lambda(2\pi)\bigr) + E_{2N+1}[h \circ \lambda_{N}],
\end{equation}
where \(E_{M}[f] \defeq \sum_{j = 1}^{M} f(\frac{j}{M}) - M \int_{0}^{1} f(u) \dd{u}\) is the Euler--Maclaurin error term.

Our task is to estimate \(E_{2N+1}[h \circ \lambda_{N}]\).
The key feature is that \(h(\lambda_{N}(u))\) has a logarithmic singularity at \(u = 0\), arising from a quadratic zero of \(\lambda_{N}(u)\):
\begin{equation}\label{eq:rescaled-eigenbranch-quadratic-zero}
    \lambda_{N}(u) = \gamma_{N} u^{2} + O(u^4), \qquad \gamma_{N} \defeq \pi^{2} \gamma \Bigl(\frac{2N + 1}{N + 1 + \eta'(0)}\Bigr)^{2}.
\end{equation}
To isolate this singularity, we decompose \(h(\lambda_{N}(u)) = F_{N}(u) + R_{N}(u)\), where \(F_{N}\) contains the singularity at \(u = 0\) but vanishes near the upper endpoint \(u = 1\), while \(R_{N}\) is smooth.
Specifically, we define
\begin{align}
    F_{N}(u) &\defeq -\chi(u) \lambda_{N}(u) \log(u^{2}), \\[1ex]
    R_{N}(u) &\defeq -\chi(u) \lambda_{N}(u) \log\Bigl(\frac{\lambda_{N}(u)}{u^{2}}\Bigr) - \bigl(1 - \chi(u)\bigr) \lambda_{N}(u) \log\bigl(\lambda_{N}(u)\bigr) - \bigl(1 - \lambda_{N}(u)\bigr) \log\bigl(1 - \lambda_{N}(u)\bigr),
\end{align}
where \(\chi(u)\) is a smooth function compactly supported in \((-1, 1)\) that equals \(1\) in a neighborhood of \(u = 0\).

We now use standard results to compute the Euler--Maclaurin error terms for \(F_{N}\) and \(R_{N}\).
The total error term will then be given by their sum:
\begin{equation}
    E_{2N+1}[h \circ \lambda_{N}] = E_{2N+1}[F_{N}] + E_{2N+1}[R_{N}].
\end{equation}

We start with the error term for \(R_{N}\).
By construction, \(R_{N}(u)\) is smooth and satisfies \(R_{N}(0) = 0\) and \(R_{N}(1) = h(\lambda(2\pi))\).
Moreover, Lemma~\ref{lem:rescaled-eigenbranch-properties}(\ref{item:rescaled-eigenbranch-symmetry}) implies that \(R_{N}^{(2r+1)}(0) = R_{N}^{(2r+1)}(1) = 0\) for every integer \(r \geq 0\).
Therefore, applying the regular Euler--Maclaurin estimate in Eq.~(1) of Ref.~\cite{Navot:FurtherExtension:1962}, we obtain
\begin{equation}\label{eq:entropy-regular-euler-maclaurin}
    E_{M}[R_{N}] = \frac{1}{2} h\bigl(\lambda(2\pi)\bigr) + o(M^{-\infty}).
\end{equation}

To compute the error term for \(F_{N}\), we note that \(\lambda_{N}(0) = 0\) and, since \(\lambda_{N}\) is even, \(\lambda_{N}^{(2r+1)}(0) = 0\) for every integer \(r \geq 0\).
Therefore, applying the singular Euler--Maclaurin estimate in Eq.~(7) of Ref.~\cite{Navot:FurtherExtension:1962}, we obtain
\begin{equation}\label{eq:entropy-singular-euler-maclaurin}
    E_{M}[F_{N}] = \frac{\zeta'(-2) \lambda_{N}^{(2)}(0)}{M^{2}} + O(M^{-4}).
\end{equation}

The formulas above describe the asymptotic behavior as \(M \to \infty\), with \(N\) held fixed.
In our application, however, \(M = 2N + 1\), so \(M\) and \(N\) become large together.
Since \(F_{N}\) and \(R_{N}\) depend on \(N\), the implied constants in the remainders \(O(M^{-4})\) and \(o(M^{-\infty})\) could, in principle, grow with \(N\) and alter the resulting order estimates.
Lemma~\ref{lem:rescaled-eigenbranch-properties}(\ref{item:rescaled-eigenbranch-uniform-bounds}) addresses this issue.
At any fixed order of the Euler--Maclaurin expansion, the remainder in \(E_{M}[R_{N}]\) is controlled by finitely many derivatives of \(R_{N}\), while the remainder in \(E_{M}[F_{N}]\) is controlled by finitely many derivatives of the smooth factor \(\chi(u) \lambda_{N}(u)\) multiplying \(\log(u^{2})\).
Together with the quadratic behavior of \(\lambda_{N}\) established in \eqref{eq:rescaled-eigenbranch-quadratic-zero}, Lemma~\ref{lem:rescaled-eigenbranch-properties}(\ref{item:rescaled-eigenbranch-uniform-bounds}) gives uniform bounds on these derivatives, so the implied constants can be chosen independently of \(N\).

Combining \eqref{eq:entropy-regular-euler-maclaurin} and \eqref{eq:entropy-singular-euler-maclaurin}, setting \(M = 2N + 1\), and using \(\zeta'(-2) = -\frac{\zeta(3)}{4\pi^{2}}\) together with \(\lambda_{N}^{(2)}(0) = 2\gamma_{N}\), we obtain
\begin{equation}
    E_{2N+1}[h \circ \lambda_{N}] = \frac{1}{2} h\bigl(\lambda(2\pi)\bigr) - \frac{\gamma \zeta(3)}{2N^{2}} + O(N^{-3}).
\end{equation}
Substituting this result into \eqref{eq:interval-entropy-euler-maclaurin} gives
\begin{equation}\label{eq:interval-entropy-integral-asymptotics}
    S(\rho_{\mathrm{int}}) = (2N + 1) \int_{0}^{1} h\bigl(\lambda_{N}(u)\bigr) \dd{u} - \frac{1}{2} h\bigl(\lambda(2\pi)\bigr) - \frac{\gamma \zeta(3)}{2N^{2}} + O(N^{-3}).
\end{equation}

It remains to evaluate the integral.
Making the change of variables \(k = q_{N}\bigl(\frac{\pi(2N + 1)}{N + 1} u\bigr)\), or equivalently \(u = \frac{(N + 1) k + \eta(k)}{\pi(2N + 1)}\), gives \(\dd{u} = \frac{N + 1 + \eta'(k)}{\pi(2N + 1)} \dd{k}\).
Substituting this into \eqref{eq:interval-entropy-integral-asymptotics} yields
\begin{equation}\label{eq:interval-entropy-asymptotics-appendix}
    S(\rho_{\mathrm{int}}) = N \bar{s} + \Delta_{h} - \frac{\gamma \zeta(3)}{2N^{2}} + O(N^{-3}),
\end{equation}
where
\begin{align}
    \bar{s} &\defeq \frac{1}{\pi} \int_{0}^{2\pi} h\bigl(\lambda(k)\bigr) \dd{k} = \frac{1}{2\pi} \int_{-\pi}^{\pi} g(k) \dd{k}, \\[1ex]
    \Delta_{h} &\defeq -\frac{1}{2} h\bigl(\lambda(2\pi)\bigr) + \frac{1}{\pi} \int_{0}^{2\pi} h\bigl(\lambda(k)\bigr) \bigl(1 + \eta'(k)\bigr) \dd{k}.
\end{align}
This is precisely the asymptotic formula \eqref{eq:interval-entropy-asymptotics} from the main text, with the \(O(1)\) boundary correction identified as \(\Delta_{h}\).

\subsection{Derivation of R\'enyi entropy \(S_{n}(\rho_{\mathrm{int}})\)}

We now derive the asymptotic formula \eqref{eq:interval-renyi-entropy-asymptotics} for the order-\(n\) R\'enyi entropy \(S_{n}(\rho_{\mathrm{int}})\) of an interval with \(N\) unit cells.
Here the R\'enyi parameter \(n\) can be any positive real number other than \(1\).
The calculation starts from \eqref{eq:interval-entropy-eigenvalue-sum}, with \(h(x)\) replaced by \(h_{n}(x) \defeq \frac{1}{1 - n} \log(x^{n} + (1 - x)^{n})\).

We first consider the simplest case, where \(n \geq 2\) is an \emph{integer}.
In this case, \(h_{n}\) is analytic in a neighborhood of the limiting spectrum of \(C_{\mathrm{int}}\), including at \(0\).
We can then apply known results from Toeplitz theory, briefly reviewed below in Appendix~\ref{sec:toeplitz-asymptotics}, to obtain
\begin{equation}\label{eq:integer-renyi-entropy-asymptotics}
    S_{n}(\rho_{\mathrm{int}}) = \Tr h_{n}(C_{\mathrm{int}}) = N \bar{s}_{n} + \Delta_{h_{n}} + o(e^{-\mu N})
\end{equation}
for some \(\mu > 0\).
In particular, the interval entropy has no algebraic finite-size corrections for integer \(n \geq 2\).

For noninteger \(n\), this simplification is lost: \(h_{n}\) is no longer analytic at \(0\), and the analysis instead parallels the von Neumann case.
Repeating the steps leading to \eqref{eq:interval-entropy-euler-maclaurin} gives
\begin{equation}\label{eq:interval-renyi-entropy-euler-maclaurin}
    S_{n}(\rho_{\mathrm{int}}) = (2N + 1) \int_{0}^{1} h_{n}\bigl(\lambda_{N}(u)\bigr) \dd{u} - h_{n}\bigl(\lambda(2\pi)\bigr) + E_{2N+1}[h_{n} \circ \lambda_{N}].
\end{equation}
As before, it remains to analyze the Euler--Maclaurin term \(E_{2N+1}[h_{n} \circ \lambda_{N}]\).

First, we write \(\lambda_{N}(u) = u^{2} \xi_{N}(u)\), where \(\xi_{N}(u)\) is smooth and \(\xi_{N}(0) = \gamma_{N}\).
Then, using the expansion \eqref{eq:renyi-entropy-series}, we decompose
\begin{equation}
    h_{n}\bigl(\lambda_{N}(u)\bigr) = \sum_{p = 1}^{\infty} F_{p,n,N}(u) + R_{n,N}(u),
\end{equation}
where
\begin{align}
    F_{p,n,N}(u) &\defeq \frac{1}{1 - n} \frac{(-1)^{p+1}}{p} \chi(u) \abs{u}^{2np} \Bigl(\frac{\xi_{N}(u)}{1 - \lambda_{N}(u)}\Bigr)^{np}, \\[1ex]
    R_{n,N}(u) &\defeq \frac{n}{1 - n} \log\bigl(1 - \lambda_{N}(u)\bigr) + \frac{1 - \chi(u)}{1 - n} \log\biggl(1 + \Bigl(\frac{\lambda_{N}(u)}{1 - \lambda_{N}(u)}\Bigr)^{n}\biggr).
\end{align}
Here \(\chi(u)\) is a smooth function compactly supported in \((-1, 1)\) that equals \(1\) in a neighborhood of \(u = 0\).
Its support is chosen sufficiently small that \(\lambda_{N}(u) < \frac{1}{4}\) whenever \(\chi(u) \neq 0\), for all sufficiently large \(N\).
This ensures that the expansion \eqref{eq:renyi-entropy-series} converges uniformly throughout the support of \(\chi\).

The function \(R_{n,N}(u)\) is smooth and satisfies \(R_{n,N}(0) = 0\) and \(R_{n,N}(1) = h_{n}(\lambda(2\pi))\).
Moreover, Lemma~\ref{lem:rescaled-eigenbranch-properties}(\ref{item:rescaled-eigenbranch-symmetry}) implies that \(R_{n,N}^{(2r+1)}(0) = R_{n,N}^{(2r+1)}(1) = 0\) for every integer \(r \geq 0\).
Applying the regular Euler--Maclaurin estimate given by Eq.~(1) of Ref.~\cite{Navot:FurtherExtension:1962} therefore gives
\begin{equation}
    E_{M}[R_{n,N}] = \frac{1}{2} h_{n}\bigl(\lambda(2\pi)\bigr) + o(M^{-\infty}).
\end{equation}
The singular Euler--Maclaurin estimate given by Eq.~(3) of Ref.~\cite{Navot:FurtherExtension:1962} similarly gives
\begin{equation}
    E_{M}[F_{p,n,N}] = \frac{1}{1 - n} \frac{(-1)^{p+1}}{p} \frac{\zeta(-2np) \bigl(\xi_{N}(0)\bigr)^{np}}{M^{2np}} + O(M^{-2np-2}).
\end{equation}
As in the von Neumann case, Lemma~\ref{lem:rescaled-eigenbranch-properties}(\ref{item:rescaled-eigenbranch-uniform-bounds}) ensures that these estimates are uniform in \(N\).

For noninteger \(n\), the \(p = 1\) term gives the leading algebraic contribution, while the terms with \(p \geq 2\) are \(o(M^{-2n})\).
Setting \(M = 2N + 1\) and using \(\xi_{N}(0) = \gamma_{N}\), we obtain
\begin{equation}
    E_{2N+1}[h_{n} \circ \lambda_{N}] = \frac{1}{2} h_{n}\bigl(\lambda(2\pi)\bigr) - \frac{(\pi^{2} \gamma)^{n} \zeta(-2n)}{n - 1} \frac{1}{N^{2n}} + o(N^{-2n}).
\end{equation}
Substituting this result into \eqref{eq:interval-renyi-entropy-euler-maclaurin} and evaluating the integral as in the von Neumann case gives
\begin{equation}\label{eq:interval-renyi-entropy-asymptotics-appendix}
    S_{n}(\rho_{\mathrm{int}}) = N \bar{s}_{n} + \Delta_{h_{n}} - \frac{(\pi^{2} \gamma)^{n} \zeta(-2n)}{n - 1} \frac{1}{N^{2n}} + o(N^{-2n}),
\end{equation}
where
\begin{align}
    \bar{s}_{n} &\defeq \frac{1}{\pi} \int_{0}^{2\pi} h_{n}\bigl(\lambda(k)\bigr) \dd{k} = \frac{1}{2\pi} \int_{-\pi}^{\pi} g_{n}(k) \dd{k}, \\[1ex]
    \Delta_{h_{n}} &\defeq -\frac{1}{2} h_{n}\bigl(\lambda(2\pi)\bigr) + \frac{1}{\pi} \int_{0}^{2\pi} h_{n}\bigl(\lambda(k)\bigr) \bigl(1 + \eta'(k)\bigr) \dd{k}.
\end{align}
This is precisely the asymptotic formula \eqref{eq:interval-renyi-entropy-asymptotics} from the main text, with the \(O(1)\) boundary correction identified as \(\Delta_{h_{n}}\).

Recall that \(\zeta(-2n) = 0\) for every integer \(n \geq 2\), consistent with the analyticity of \(h_{n}(x)\) and the absence of an algebraic correction.
Also, note that the limit \(n \to 1\) recovers the coefficient \(-\frac{\gamma \zeta(3)}{2N^{2}}\) found for the von Neumann entropy.

\subsection{Toeplitz asymptotics}\label{sec:toeplitz-asymptotics}

In this subsection, we show that if \(f\) is analytic in a neighborhood containing the ranges of the eigenvalue bands \(\lambda_{\pm}(k)\), then
\begin{equation}
    \Tr f(C_{\mathrm{int}}) = N \bar{f} + \Delta_{f} + O(e^{-\mu N})
\end{equation}
for some \(\mu > 0\), where \(\bar{f}\) and \(\Delta_{f}\) are constants independent of \(N\).
In particular, \(\Tr f(C_{\mathrm{int}})\) has no algebraic finite-size corrections.
This result gives the exponential remainder in the integer-order R\'enyi entropy expansion \eqref{eq:integer-renyi-entropy-asymptotics}, as well as in the von Neumann entropy expansions for the modified correlation matrices considered in Appendix~\ref{app:fine-tuning}.

To connect this statement with standard Toeplitz theory, recall that \(C_{\mathrm{int}}\) is an \(N \times N\) block matrix with block entries
\begin{equation}\label{eq:toeplitz-symbol-representation}
    (C_{\mathrm{int}})_{xy} = C_{x-y} = \frac{1}{2\pi} \int_{-\pi}^{\pi} C(k) e^{-i k (x-y)} \dd{k}, \qquad 1 \leq x, y \leq N.
\end{equation}
Since these entries depend only on the difference \(x - y\), \(C_{\mathrm{int}}\) is a \emph{block Toeplitz matrix}.
The matrix-valued function \(C(k)\) is conventionally called its \emph{symbol}, and we write \(T_{N}(C) = C_{\mathrm{int}}\).

Toeplitz matrices have been studied extensively, and we refer the reader to Refs.~\cite{Bottcher:IntroductionLarge:1999,Bottcher:AnalysisToeplitz:2006} for comprehensive mathematical treatments and to Ref.~\cite{Deift:ToeplitzMatrices:2013} for a survey of applications in statistical physics.
Here, we briefly review the results most relevant to our entropy calculation and clarify what can already be inferred from the existing theory.
We restrict our discussion to Hermitian symbols.

Of particular relevance here is the large-\(N\) behavior of trace formulas of the form \(\Tr f\bigl(T_{N}(C)\bigr)\).
The available results depend on the regularity of both \(f\) and \(C(k)\).
At the most basic level, a version of the first Szeg\H{o} limit theorem \cite[Corollary~6.26]{Bottcher:IntroductionLarge:1999} gives
\begin{equation}
    \Tr f\bigl(T_{N}(C)\bigr) = \frac{N}{2\pi} \int_{-\pi}^{\pi} \Tr f\bigl(C(k)\bigr) \dd{k} + o(N),
\end{equation}
when \(f\) and \(C(k)\) are merely continuous.
This is already enough to determine the leading extensive term in either \eqref{eq:interval-entropy-asymptotics} or \eqref{eq:integer-renyi-entropy-asymptotics}.

If \(f\) is analytic in a neighborhood containing the spectrum of \(C(k)\) and \(C(k)\) satisfies mild regularity conditions, Widom's trace formula \cite[Theorem~6.2]{Widom:AsymptoticBehavior:1976} gives
\begin{equation}\label{eq:widom-trace-asymptotics}
    \Tr f\bigl(T_{N}(C)\bigr) = \frac{N}{2\pi} \int_{-\pi}^{\pi} \Tr f\bigl(C(k)\bigr) \dd{k} + \Delta_{f} + o(1),
\end{equation}
where \(\Delta_{f}\) is an \(N\)-independent constant.
Stronger regularity assumptions on \(C(k)\) give correspondingly stronger control of the remainder.
In particular, the higher-order trace asymptotics of Section~10.91 of Ref.~\cite{Bottcher:AnalysisToeplitz:2006} imply that sufficiently rapid decay of the Fourier coefficients of \(C(k)\) yields correspondingly rapid algebraic decay of the remainder.
Since \(C(k)\) in our case has a finite Fourier series, these results imply that the remainder is \(o(N^{-\infty})\).
Theorem~5 of Ref.~\cite{Gutierrez-Gutierrez:FunctionsBanded:2007} further improves this to \(o(e^{-\mu N})\) for some \(\mu > 0\), establishing the claimed exponential remainder.

When \(f\) is not analytic in a neighborhood of the spectrum of \(C(k)\), the available results appear to be considerably more limited.
Although the first Szeg\H{o} limit theorem still determines the leading extensive term, we are not aware of general results that determine the subextensive terms relevant to the CMI.
For the von Neumann entropy in particular, related work has circumvented this nonanalyticity by introducing a regulator that makes the entropy function analytic in a neighborhood of the spectrum~\cite{Its:EntanglementEntropy:2008,Its:FisherHartwigFormula:2009}.
Toeplitz asymptotic results can then be applied to the regularized problem, after which the regulator is removed.
Even with this regularization, however, it is again not clear how to obtain the subextensive terms that enter the CMI in our setting.
We therefore instead exploit the explicit eigenvalue structure of \(C_{\mathrm{int}}\) to extract the subextensive terms directly using the Euler--Maclaurin analysis above.

\section{Fine-tuning}\label{app:fine-tuning}

In this appendix, we show how the correlation matrix changes when the two-site modes in \eqref{eq:two-site-mode} are replaced by the more general choice \(s_{j} = \cos(\theta) c_{j} + \sin(\theta) c_{j+1}\), with \(0 < \theta < \frac{\pi}{2}\).
In particular, we show that the zero of the lower eigenvalue band \(\lambda_{-}(k)\) is lifted unless \(\theta = \frac{\pi}{4}\).

The derivation of the correlation matrix follows the same steps as in Appendix~\ref{app:correlation-matrix}, with \(X\) and \(L\) replaced by
\begin{equation}
    X = \begin{pmatrix}
        \cos^{2}(\frac{\epsilon}{2}) - \sin^{2}(\frac{\epsilon}{2}) \cos(2\theta) & -\sin^{2}(\frac{\epsilon}{2}) \sin(2\theta) \\[2ex]
        -\sin^{2}(\frac{\epsilon}{2}) \sin(2\theta) & \cos^{2}(\frac{\epsilon}{2}) + \sin^{2}(\frac{\epsilon}{2}) \cos(2\theta)
    \end{pmatrix}, \qquad L = \sin(\epsilon) \begin{pmatrix}
        \cos(\theta) \\[1ex]
        \sin(\theta)
    \end{pmatrix}.
\end{equation}
The resulting momentum-space correlation matrix is
\begin{equation}\label{eq:detuned-correlation-symbol}
    C(k) = \begin{pmatrix}
        a_{+}(k) & b(k) \\[1ex]
        \bar{b}(k) & a_{-}(k)
    \end{pmatrix},
\end{equation}
where
\begin{align}
    \label{eq:detuned-correlation-symbol-diagonal} a_{\pm}(k) &= \frac{1}{2} \sin^{2}(\epsilon) - \frac{1}{4} \sin^{4}(\epsilon) \sin^{2}(2\theta) \Bigl(1 \pm \tan^{2}\Bigl(\frac{\epsilon}{2}\Bigr) \cos(2\theta)\Bigr) \cos^{2}\Bigl(\frac{k}{2}\Bigr), \\[0.5ex]
    \label{eq:detuned-correlation-symbol-offdiagonal} b(k) &= \frac{1}{4} \sin^{2}(\epsilon) \sin(2\theta) (1 + e^{i k}) \Bigl(\cos(\epsilon) + \sin^{4}\Bigl(\frac{\epsilon}{2}\Bigr) \sin^{2}(2\theta) (1 + e^{i k})\Bigr).
\end{align}

The two eigenvalue bands of \(C(k)\) are
\begin{equation}\label{eq:detuned-eigenvalue-bands}
    \lambda_{\pm}(k) = \frac{1}{2} \sin^{2}(\epsilon) \Biggl(1 - \frac{1}{2} \sin^{2}(\epsilon) \, r(k)^{2} \pm r(k) \sqrt{\cos^{2}(\epsilon) + \frac{1}{4} \sin^{4}(\epsilon) r(k)^{2}}\Biggr), \qquad r(k) \defeq \sin(2\theta) \abs[\Big]{\cos\Bigl(\frac{k}{2}\Bigr)}.
\end{equation}
To bound these eigenvalues, write \(\lambda_{\pm}(k) = \Lambda(\pm r(k))\), where
\begin{equation}
    \Lambda(x) \defeq \frac{1}{2} \sin^{2}(\epsilon) \Biggl(1 - \frac{1}{2} \sin^{2}(\epsilon) \, x^{2} + x \sqrt{\cos^{2}(\epsilon) + \frac{1}{4} \sin^{4}(\epsilon) x^{2}}\Biggr).
\end{equation}
A direct calculation shows that \(\Lambda\) is strictly increasing on the whole real line, and that \(0 < \Lambda(x) < \frac{1}{2}\) for \(\abs{x} < 1\).
If \(\theta \neq \frac{\pi}{4}\), then \(\abs{r(k)} \leq \sin(2\theta) < 1\), and therefore \(0 < \Lambda(-\sin(2\theta)) \leq \lambda_{\pm}(k) < \frac{1}{2}\).
In particular, both eigenvalue bands are (uniformly in \(k\)) bounded away from zero.

\section{Alternative representations of \(\mathcal{E}_{j}\)}\label{app:channel-representations}

In this appendix, we derive two alternative representations of the two-site channel
\begin{equation}\label{eq:local-channel-appendix}
    \mathcal{E}_{j}(\rho) = \Tr_{d_{j}}\Bigl[U_{j} \bigl(\rho \otimes I_{d_{j}}/2\bigr) U_{j}^{\dagger}\Bigr], \qquad U_{j} = e^{\epsilon(s_{j}^{\dagger} d_{j} - d_{j}^{\dagger} s_{j})},
\end{equation}
which will be useful in later discussions.
These representations are summarized by the following two propositions.
For the first, it is necessary to introduce the two Majorana operators associated with \(s_{j}\):
\begin{equation}
    \gamma_{j}^{x} \defeq s_{j}^{\dagger} + s_{j}, \qquad \gamma_{j}^{y} \defeq i(s_{j}^{\dagger} - s_{j}).
\end{equation}

\begin{proposition}\label{prop:channel-random-unitary}
    For \(0 \leq \epsilon < \frac{\pi}{2}\), define \(p \defeq \sin^{2}(\frac{\epsilon}{2})\).
    Then the channel \eqref{eq:local-channel-appendix} can be written in random-unitary form,
    \begin{equation}\label{eq:channel-random-unitary}
        \mathcal{E}_{j}(\rho) = \sum_{a,b = 0}^{1} p_{ab} \, \gamma_{j}^{ab} \rho (\gamma_{j}^{ab})^{\dagger}, \qquad p_{ab} \defeq p^{a+b} (1 - p)^{2-a-b}, \qquad \gamma_{j}^{ab} \defeq (\gamma_{j}^{y})^{b} (\gamma_{j}^{x})^{a}.
    \end{equation}
\end{proposition}

\begin{proof}
    Define the ancilla Majorana operators \(\eta_{j}^{x} \defeq d_{j}^{\dagger} + d_{j}\) and \(\eta_{j}^{y} \defeq i(d_{j}^{\dagger} - d_{j})\).
    Then
    \begin{equation}
        U_{j} = e^{\frac{\epsilon}{2}(\gamma_{j}^{x} \eta_{j}^{x} + \gamma_{j}^{y} \eta_{j}^{y})} = \bigl(\sqrt{1 - p} + \sqrt{p} \, \gamma_{j}^{x} \eta_{j}^{x}\bigr) \bigl(\sqrt{1 - p} + \sqrt{p} \, \gamma_{j}^{y} \eta_{j}^{y}\bigr).
    \end{equation}
    Inserting this into \eqref{eq:local-channel-appendix} and tracing out the ancilla gives
    \begin{equation}
        \mathcal{E}_{j}(\rho) = (1 - p)^{2} \rho + p(1 - p) \bigl(\gamma_{j}^{x} \rho \gamma_{j}^{x} + \gamma_{j}^{y} \rho \gamma_{j}^{y}\bigr) + p^{2} \gamma_{j}^{y} \gamma_{j}^{x} \rho \gamma_{j}^{x} \gamma_{j}^{y},
    \end{equation}
    which is exactly \eqref{eq:channel-random-unitary}.
\end{proof}

\begin{proposition}\label{prop:channel-markovian}
    For \(0 \leq \epsilon < \frac{\pi}{2}\), define \(t \defeq -\log\bigl(\cos^{2}(\epsilon)\bigr)\).
    Then the channel \eqref{eq:local-channel-appendix} can be written in Markovian form,
    \begin{equation}\label{eq:channel-markovian}
        \mathcal{E}_{j}(\rho) = e^{t \mathcal{L}_{j}}(\rho), \qquad \mathcal{L}_{j}(\rho) \defeq \frac{1}{2} \bigl(s_{j} \rho s_{j}^{\dagger} + s_{j}^{\dagger} \rho s_{j} - \rho\bigr).
    \end{equation}
\end{proposition}

\begin{proof}
    We show that \eqref{eq:channel-markovian} is equivalent to \eqref{eq:channel-random-unitary}.
    Using the Majorana operators introduced above, the Lindbladian \(\mathcal{L}_{j}\) can be written as
    \begin{equation}
        \mathcal{L}_{j}(\rho) = \frac{1}{4} \bigl(\gamma_{j}^{x} \rho \gamma_{j}^{x} + \gamma_{j}^{y} \rho \gamma_{j}^{y} - 2\rho\bigr),
    \end{equation}
    or equivalently,
    \begin{equation}
        \mathcal{L}_{j} = \frac{1}{4} \Ad_{\gamma_{j}^{x}} + \frac{1}{4} \Ad_{\gamma_{j}^{y}} - \frac{1}{2}.
    \end{equation}
    Since \(\Ad_{\gamma_{j}^{x}}^{2} = \Ad_{\gamma_{j}^{y}}^{2} = 1\) and \([\Ad_{\gamma_{j}^{x}}, \Ad_{\gamma_{j}^{y}}] = 0\), we can explicitly evaluate the exponential as
    \begin{equation}
        e^{t \mathcal{L}_{j}} = \biggl(\frac{1 + e^{-t/2}}{2} + \frac{1 - e^{-t/2}}{2} \Ad_{\gamma_{j}^{x}}\biggr) \biggl(\frac{1 + e^{-t/2}}{2} + \frac{1 - e^{-t/2}}{2} \Ad_{\gamma_{j}^{y}}\biggr).
    \end{equation}
    Since \(e^{-t} = \cos^{2}(\epsilon)\), we have \(p = \frac{1 - e^{-t/2}}{2}\), and hence
    \begin{equation}
        e^{t \mathcal{L}_{j}} = \Bigl((1 - p) + p \Ad_{\gamma_{j}^{x}}\Bigr) \Bigl((1 - p) + p \Ad_{\gamma_{j}^{y}}\Bigr),
    \end{equation}
    which, upon expansion, is precisely \eqref{eq:channel-random-unitary}.
\end{proof}

\section{Lindbladian analog}\label{app:lindbladian-model}

In this appendix, we provide the details of the Lindbladian analog discussed briefly in the main text.
Let \(\mathcal{N} = e^{t \mathcal{L}}\), where \(\mathcal{L}\) is the time-independent Lindbladian
\begin{equation}
    \mathcal{L}(\rho) = \sum_{j \in \mathbb{Z}} \mathcal{L}_{j}(\rho), \qquad \mathcal{L}_{j}(\rho) = \frac{1}{2} \bigl(s_{j} \rho s_{j}^{\dagger} + s_{j}^{\dagger} \rho s_{j} - \rho\bigr).
\end{equation}
Starting from the vacuum \(\rho_{0}\), define \(\rho(t) = \mathcal{N}(\rho_{0})\).
The state \(\rho(t)\) is invariant under single-site translations.
Its correlation matrix therefore depends only on the separation, \(C_{j-k} \defeq \Tr[\rho(t) c_{j}^{\dagger} c_{k}]\).
Differentiating with respect to time gives
\begin{equation}
    \dot{C}_{j-k} = \Tr\bigl[\mathcal{L}\bigl(\rho(t)\bigr) c_{j}^{\dagger} c_{k}\bigr] = X_{j-k} - \sum_{\ell \in \mathbb{Z}} \bigl(X_{j-\ell} C_{\ell-k} + C_{j-\ell} X_{\ell-k}\bigr),
\end{equation}
where \(X_{j-k} \defeq \frac{1}{2} \delta_{j,k} + \frac{1}{4} (\delta_{j+1,k} + \delta_{j,k+1})\).

This equation is most easily solved in momentum space.
Define \(C(k) \defeq \sum_{j \in \mathbb{Z}} e^{i k j} C_{j}\) and \(X(k) \defeq \sum_{j \in \mathbb{Z}} e^{i k j} X_{j}\).
The preceding equation reduces to the scalar equation \(\dot{C}(k) = -2X(k) C(k) + X(k)\).
Since \(C(k) = 0\) initially and \(X(k) = \cos^{2}(\frac{k}{2})\), the solution is
\begin{equation}
    C(k) = \frac{1}{2} \bigl(1 - e^{-2t X(k)}\bigr) = \frac{1}{2} \bigl(1 - e^{-2t \cos^{2}(\frac{k}{2})}\bigr).
\end{equation}
Fourier transforming back to real space gives \(C_{j-k} = \frac{1}{2} (\delta_{j,k} - e^{-t} I_{j-k}(-t))\), where \(I_{\nu}\) is the modified Bessel function of the first kind.
Unlike the finite-depth example in \eqref{eq:real-space-correlation-blocks}, \(C_{j-k}\) is not \emph{strictly} finite-range, though it does decay exponentially as \(\abs{j-k} \to \infty\).

We next consider the entropy of an interval of \(N\) two-site unit cells, or equivalently \(2N\) sites.
This requires the spectrum of the \(2N \times 2N\) \emph{scalar} Toeplitz matrix \(C_{\mathrm{int}}\) with entries \((C_{\mathrm{int}})_{jk} = C_{j-k}\) for \(1 \leq j, k \leq 2N\).
The following lemma gives the spectral structure needed below.

\begin{lemma}\label{lem:lindbladian-spectrum}
    For \(N\) sufficiently large, the eigenvalues of \(C_{\mathrm{int}}\) are all distinct and can be written as \(\sigma_{j} = \lambda(k_{j})\) for \(j \in \{1, 2, \ldots, 2N\}\), where
    \begin{equation}
        \lambda(k) \defeq \frac{1}{2} \bigl(1 - e^{-2t \sin^{2}(\frac{k}{2})}\bigr)
    \end{equation}
    and \(k_{j} \in (0, \pi)\) is the unique solution of
    \begin{equation}\label{eq:lindbladian-quantization}
        (2N + 1) k_{j} + \eta(k_{j}) = \pi j + \delta_{N}(j).
    \end{equation}
    Here, \(\eta(k)\) is odd, analytic, and \(2\pi\)-periodic, while \(\delta_{N}(j) = o(N^{-\infty})\) uniformly with respect to \(j\).
\end{lemma}

\begin{proof}
    We first observe that \(\lambda(k) = C(k - \pi)\).
    If \(\widetilde{C}_{\mathrm{int}}\) is the Toeplitz matrix generated by \(\lambda(k)\) in place of \(C(k)\), then \(\widetilde{C}_{\mathrm{int}} = U C_{\mathrm{int}} U^{-1}\), where \(U = \operatorname{diag}(-1, 1, -1, \ldots)\).
    Thus, \(\widetilde{C}_{\mathrm{int}}\) and \(C_{\mathrm{int}}\) have the same spectrum.

    The symbol \(\lambda(k)\) belongs to the class of ``smooth simple-loop'' symbols considered in Ref.~\cite{Bogoya:EigenvaluesHermitian:2015}.
    Our lemma is therefore a direct specialization of Theorem~2.1 of that reference, which gives the distinctness of the eigenvalues, the quantization condition \eqref{eq:lindbladian-quantization}, and the uniqueness of each \(k_{j} \in (0, \pi)\).
    Since \(\lambda(k)\) is smooth, the theorem may be applied with arbitrarily large regularity index, giving \(\delta_{N}(j) = o(N^{-\infty})\) uniformly in \(j\).

    For the present symbol, the phase function appearing in the theorem is
    \begin{equation}
        \eta(k) = \frac{t \sin(k)}{\pi} \int_{0}^{\pi} F\bigl(t \cos(u) - t \cos(k)\bigr) \dd{u}, \qquad F(x) \defeq \frac{1}{x} \log\biggl(\frac{e^{x} - 1}{x}\biggr).
    \end{equation}
    This representation shows that \(\eta(k)\) is odd, analytic, and \(2\pi\)-periodic.
\end{proof}

We can now proceed as in Appendix~\ref{app:interval-entropy-asymptotics}.
Define \(\phi_{N}(k) \defeq k + \frac{\eta(k)}{2N + 1}\) and invert the quantization condition to obtain
\begin{equation}
    S(\rho_{\mathrm{int}}) = \sum_{j = 1}^{2N} h\biggl(\lambda_{N}\biggl(\frac{j}{2N + 1} + \frac{\delta_{N}(j)}{\pi(2N + 1)}\biggr)\biggr), \qquad \lambda_{N}(u) \defeq \lambda\bigl(\phi_{N}^{-1}(\pi u)\bigr).
\end{equation}
Since \(h \circ \lambda_{N}\) has a continuous first derivative bounded uniformly in \(N\), say by \(K > 0\), we have
\begin{equation}
    \abs[\Bigg]{S(\rho_{\mathrm{int}}) - \sum_{j = 1}^{2N} h\biggl(\lambda_{N}\Bigl(\frac{j}{2N + 1}\Bigr)\biggr)} \leq \frac{K}{\pi(2N + 1)} \sum_{j = 1}^{2N} \abs{\delta_{N}(j)} \leq \frac{K}{\pi} \max_{j} \abs{\delta_{N}(j)} = o(N^{-\infty}).
\end{equation}
Thus, the error in the quantization condition contributes only a superpolynomially small correction to the entropy.

The remaining sum can be evaluated by adapting the methods of Appendix~\ref{app:interval-entropy-asymptotics}.
The result is
\begin{equation}
    S(\rho_{\mathrm{int}}) = N \bar{s} + \Delta_{h} - \frac{t \zeta(3)}{32N^{2}} + O(N^{-3}),
\end{equation}
where \(\bar{s} \defeq \frac{2}{\pi} \int_{0}^{\pi} h(\lambda(k)) \dd{k}\) and \(\Delta_{h} \defeq -\frac{1}{2} h(\lambda(\pi)) + \frac{1}{\pi} \int_{0}^{\pi} h(\lambda(k)) (1 + \eta'(k)) \dd{k}\).

\section{Qubit model}\label{app:qubit-model}

In this appendix, we explicitly construct the qubit example described in the main text and prove its equivalence to the fermionic model.
More precisely, we construct a depth-two local channel \(\mathcal{E}^{\mathrm{q}}\) on an infinite chain of qubits, obtained by replacing each fermionic gate \(\mathcal{E}_{j}\) by a corresponding qubit gate \(\mathcal{E}_{j}^{\mathrm{q}}\), such that for any fixed contiguous regions \(A\), \(B\), and \(C\),
\begin{equation}\label{eq:qubit-fermion-cmi-equivalence}
    \CMI{A}{C}{B}_{\mathcal{E}^{\mathrm{q}}(\rho_{0}^{\mathrm{q}})} = \CMI{A}{C}{B}_{\mathcal{E}(\rho_{0})},
\end{equation}
where \(\rho_{0}\) and \(\rho_{0}^{\mathrm{q}}\) are the fermionic vacuum and qubit product state, respectively.

Because the channels have fixed finite depth, the only gates that affect \(\CMI{A}{C}{B}\) are those within the causal cone of \(ABC\).
Therefore, for the purposes of showing \eqref{eq:qubit-fermion-cmi-equivalence}, we need only retain these gates and the sites on which they act.
This amounts to truncating the infinite chain to a \emph{finite} open chain surrounding \(ABC\), whose gate structure is shown in Fig.~\ref{fig:finite-chain-causal-cone}.

\begin{figure}
    \centering
    \input{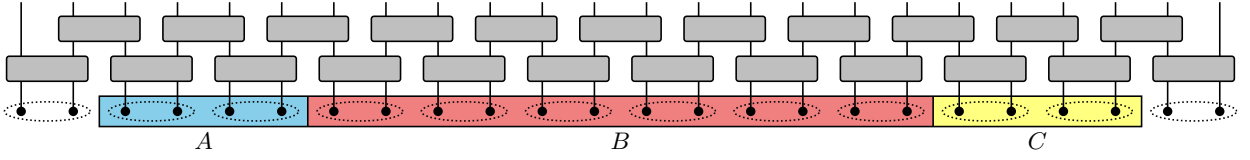}
    \caption{\label{fig:finite-chain-causal-cone}
        Channel structure of a finite open chain containing the causal cone of \(ABC\).
    }
\end{figure}

For concreteness, we take this chain to consist of \(2N\) sites indexed by \(\{0, 1, \ldots, 2N - 1\}\).
Let \(\rho_{0} = \ketbra{\Omega}{\Omega}\) denote the fermionic vacuum on the open chain.
The fermionic channel \(\mathcal{E}\) now consists of \(2N - 1\) \emph{distinct} gates,
\begin{equation}\label{eq:finite-chain-channel}
    \mathcal{E} = \mathcal{E}_{1} \circ \mathcal{E}_{3} \circ \cdots \circ \mathcal{E}_{2N-3} \circ \mathcal{E}_{0} \circ \mathcal{E}_{2} \circ \cdots \circ \mathcal{E}_{2N-2}.
\end{equation}
By Proposition~\ref{prop:channel-random-unitary}, these local gates have the random-unitary representation
\begin{equation}\label{eq:channel-random-unitary-qubit-derivation}
    \mathcal{E}_{j}(\rho) = \sum_{a,b = 0}^{1} p_{ab} \, \gamma_{j}^{ab} \rho (\gamma_{j}^{ab})^{\dagger}, \qquad p_{ab} = p^{a+b} (1 - p)^{2-a-b},
\end{equation}
where \(p = \sin^{2}(\frac{\epsilon}{2})\) and \(\gamma_{j}^{ab} = (\gamma_{j}^{y})^{b} (\gamma_{j}^{x})^{a}\).
Here \(\gamma_{j}^{x} = s_{j}^{\dagger} + s_{j}\) and \(\gamma_{j}^{y} = i(s_{j}^{\dagger} - s_{j})\) are the Majorana operators associated with the mode \(s_{j}\).

We introduce \(2N\) qubits on the same chain with Pauli operators \(X_{j}\), \(Y_{j}\), and \(Z_{j}\), and use the local basis \(Z \ket{0} = \ket{0}\), \(Z \ket{1} = -\ket{1}\).
The initial qubit state is \(\rho_{0}^{\mathrm{q}} = \bigotimes_{j = 0}^{2N-1} \ketbra{0}{0}\).
The Jordan--Wigner transformation is the unitary map \(J\) from the fermionic Fock space to the qubit Hilbert space defined by
\begin{equation}
    J\Bigl[(c_{0}^{\dagger})^{n_{0}} (c_{1}^{\dagger})^{n_{1}} \cdots (c_{2N-1}^{\dagger})^{n_{2N-1}} \ket{\Omega}\Bigr] = \ket{n_{0}} \otimes \ket{n_{1}} \otimes \cdots \otimes \ket{n_{2N-1}}.
\end{equation}
For our purposes, the relevant identities are
\begin{equation}
    J \rho_{0} J^{\dagger} = \rho_{0}^{\mathrm{q}}, \qquad J c_{j} J^{\dagger} = \frac{1}{2} (X_{j} + iY_{j}) P_{j-1},
\end{equation}
where \(P_{j} \defeq \prod_{k = 0}^{j} Z_{k}\) are the Jordan--Wigner parity strings.

We now construct the local qubit gates \(\mathcal{E}_{j}^{\mathrm{q}}\).
Consider the Jordan--Wigner images of the Majorana operators \(\gamma_{j}^{x}\) and \(\gamma_{j}^{y}\).
Each consists of a local two-site operator together with the same parity string \(P_{j-1}\):
\begin{alignat}{2}
    J \gamma_{j}^{x} J^{\dagger} &= G_{j}^{x} P_{j-1}, &\qquad G_{j}^{x} &\defeq \frac{1}{\sqrt{2}} \bigl(X_{j} + Z_{j} X_{j+1}\bigr), \\[1ex]
    J \gamma_{j}^{y} J^{\dagger} &= G_{j}^{y} P_{j-1}, &\qquad G_{j}^{y} &\defeq \frac{1}{\sqrt{2}} \bigl(Y_{j} + Z_{j} Y_{j+1}\bigr).
\end{alignat}
It follows that the composite unitary \(\gamma_{j}^{ab}\) has Jordan--Wigner image
\begin{equation}\label{eq:qubit-string-compensated-unitaries}
    J \gamma_{j}^{ab} J^{\dagger} = R_{j}^{ab} P_{j}^{a+b}, \qquad R_{j}^{ab} \defeq (G_{j}^{y})^{b} (G_{j}^{x})^{a} Z_{j}^{a+b}.
\end{equation}
Here, we included the factor \(Z_{j}^{a+b}\) in \(R_{j}^{ab}\) to compensate for extending the parity string from \(P_{j-1}\) to \(P_{j}\).
This convention is useful because the parity string associated with \(\mathcal{E}_{j}\) can now be commuted through all previously applied gates.
Then, since the parity string acts trivially on \(\rho_{0}^{\mathrm{q}}\), it does not affect the action of the circuit, as we will see below.
Motivated by this, we define the corresponding qubit gate \(\mathcal{E}_{j}^{\mathrm{q}}\) by omitting these strings, leaving only the local unitaries \(R_{j}^{ab}\):
\begin{equation}
    \mathcal{E}_{j}^{\mathrm{q}}(\rho) \defeq \sum_{a,b = 0}^{1} p_{ab} \, R_{j}^{ab} \rho (R_{j}^{ab})^{\dagger}.
\end{equation}
The full qubit channel \(\mathcal{E}^{\mathrm{q}}\) is obtained by replacing each gate \(\mathcal{E}_{j}\) in \eqref{eq:finite-chain-channel} by \(\mathcal{E}_{j}^{\mathrm{q}}\).
We can now state our main equivalence result.

\begin{proposition}
    Let \(\mathcal{E}_{j_{1}}, \mathcal{E}_{j_{2}}, \ldots, \mathcal{E}_{j_{k}}\) be a sequence of distinct fermionic gates.
    Then
    \begin{equation}\label{eq:distinct-gate-qubit-equivalence}
        \mathcal{E}_{j_{k}}^{\mathrm{q}} \circ \cdots \circ \mathcal{E}_{j_{1}}^{\mathrm{q}}(\rho_{0}^{\mathrm{q}}) = J\Bigl(\mathcal{E}_{j_{k}} \circ \cdots \circ \mathcal{E}_{j_{1}}(\rho_{0})\Bigr) J^{\dagger}.
    \end{equation}
\end{proposition}

\begin{proof}
    First, we note that
    \begin{equation}
        P_{j} R_{\ell}^{ab} = \begin{cases}
            R_{\ell}^{ab} P_{j}, & j < \ell, \\[1ex]
            (-1)^{a+b} R_{\ell}^{ab} P_{j}, & j > \ell.
        \end{cases}
    \end{equation}
    Consequently, for \(j \neq \ell\),
    \begin{equation}\label{eq:parity-channel-covariance}
        P_{j} \mathcal{E}_{\ell}^{\mathrm{q}}(\,\cdot\,) P_{j} = \mathcal{E}_{\ell}^{\mathrm{q}}\bigl(P_{j} (\,\cdot\,) P_{j}\bigr).
    \end{equation}

    We now prove \eqref{eq:distinct-gate-qubit-equivalence} by induction on \(k\).
    For any single gate \(\mathcal{E}_{j}\), using \eqref{eq:qubit-string-compensated-unitaries} gives
    \begin{equation}
        J \mathcal{E}_{j}(\rho_{0}) J^{\dagger} = \sum_{a,b = 0}^{1} p_{ab} \, R_{j}^{ab} P_{j}^{a+b} \rho_{0}^{\mathrm{q}} P_{j}^{a+b} (R_{j}^{ab})^{\dagger}.
    \end{equation}
    Since \(P_{j} \rho_{0}^{\mathrm{q}} P_{j} = \rho_{0}^{\mathrm{q}}\), the parity strings drop out, so \(J \mathcal{E}_{j}(\rho_{0}) J^{\dagger} = \mathcal{E}_{j}^{\mathrm{q}}(\rho_{0}^{\mathrm{q}})\).
    Thus, \eqref{eq:distinct-gate-qubit-equivalence} holds for \(k = 1\).

    Now suppose inductively that \eqref{eq:distinct-gate-qubit-equivalence} holds for every sequence of \(k - 1\) distinct gates.
    Using \eqref{eq:qubit-string-compensated-unitaries} and the inductive hypothesis gives
    \begin{equation}\label{eq:distinct-gate-equivalence-induction}
        J\Bigl(\mathcal{E}_{j_{k}} \circ \cdots \circ \mathcal{E}_{j_{1}}(\rho_{0})\Bigr) J^{\dagger} = \sum_{a,b = 0}^{1} p_{ab} \, R_{j_{k}}^{ab} P_{j_{k}}^{a+b} \Bigl(\mathcal{E}_{j_{k-1}}^{\mathrm{q}} \circ \cdots \circ \mathcal{E}_{j_{1}}^{\mathrm{q}}(\rho_{0}^{\mathrm{q}})\Bigr) P_{j_{k}}^{a+b} (R_{j_{k}}^{ab})^{\dagger}.
    \end{equation}
    Since the gates are distinct, \(j_{k} \notin \{j_{1}, \ldots, j_{k-1}\}\).
    Repeated use of \eqref{eq:parity-channel-covariance}, together with the invariance of \(\rho_{0}^{\mathrm{q}}\) under \(P_{j_{k}}\), therefore gives
    \begin{equation}
        P_{j_{k}}^{a+b} \Bigl(\mathcal{E}_{j_{k-1}}^{\mathrm{q}} \circ \cdots \circ \mathcal{E}_{j_{1}}^{\mathrm{q}}(\rho_{0}^{\mathrm{q}})\Bigr) P_{j_{k}}^{a+b} = \mathcal{E}_{j_{k-1}}^{\mathrm{q}} \circ \cdots \circ \mathcal{E}_{j_{1}}^{\mathrm{q}}(\rho_{0}^{\mathrm{q}}).
    \end{equation}
    Substituting this into \eqref{eq:distinct-gate-equivalence-induction} yields
    \begin{equation}
        J\Bigl(\mathcal{E}_{j_{k}} \circ \cdots \circ \mathcal{E}_{j_{1}}(\rho_{0})\Bigr) J^{\dagger} = \mathcal{E}_{j_{k}}^{\mathrm{q}} \circ \cdots \circ \mathcal{E}_{j_{1}}^{\mathrm{q}}(\rho_{0}^{\mathrm{q}}),
    \end{equation}
    which closes the induction.
\end{proof}

Applying the proposition to the complete sequence of gates in \eqref{eq:finite-chain-channel} gives \(J \mathcal{E}(\rho_{0}) J^{\dagger} = \mathcal{E}^{\mathrm{q}}(\rho_{0}^{\mathrm{q}})\).
Since \(AB\), \(BC\), \(B\), and \(ABC\) are all contiguous intervals, the Jordan--Wigner equivalence descends to the corresponding reduced states, which therefore have identical entropies.
It follows that \eqref{eq:qubit-fermion-cmi-equivalence} holds.

\end{document}